\documentclass{CSML}
\pdfoutput=1

\usepackage{lastpage}
\newcommand{\dOi}{13(3:9)2017}
\lmcsheading{}{1--\pageref{LastPage}}{}{}%
{Nov.~24, 2016}{Aug.~15, 2017}{}

\usepackage{hyperref}
\hypersetup{hidelinks}

\usepackage{algorithmic}
\usepackage{algorithm}
\usepackage{csquotes}
\usepackage{stmaryrd}
\usepackage{bussproofs}
\usepackage{mathrsfs}  
\usepackage{mathtools}

\theoremstyle{plain}
\newtheorem{problem}{Problem}

\newtheorem{definition}{Definition}
\newtheorem{lemma}{Lemma}
\newtheorem{corollary}{Corollary}
\newtheorem{example}{Example}
\newtheorem{theorem}{Theorem}
\newtheorem{remark}{Remark}

\def\CC{\mathbb{C}}

\def\VV{\mathbb{V}}
\def\PP{\mathbb{P}}
\def\TT{\mathbb{T}}
\def\NN{\mathbb{N}}

\def\calO{\mathcal{O}}
\def\calC{\mathcal{C}}

\def\calD{\mathcal{D}}

\DeclareMathOperator{\dom}{dom}
\DeclareMathOperator{\depth}{depth}

\DeclareMathOperator{\src}{src}

\DeclareMathOperator{\tgt}{tgt}

\newcommand{\dotleq}[0]{\mathrel{\dot{\leq}}}

\DeclareMathOperator{\Var}{Var}

\def\AX{\ensuremath{(\text{Ax})}}
\def\OM{\ensuremath{(\omega)}}
\def\II{\ensuremath{(\text{$\cap$I})}}
\def\SUB{\ensuremath{(\leq)}}

\def\AI{\ensuremath{(\text{$\to$I})}}
\def\AE{\ensuremath{(\text{$\to$E})}}

\newcommand{\FORK}[1]{\STATE \textbf{fork} #1 \begin{ALC@g}}
\newcommand{\ENDFORK}{\end{ALC@g}}

\newenvironment{ALC@either}{\begin{ALC@g}}{\end{ALC@g}}
\newcommand{\EITHER}[0]{\STATE \textbf{either} \begin{ALC@either}}
\newcommand{\EITHEROR}[0]{\end{ALC@either} \STATE \textbf{or} \begin{ALC@either}}
\newcommand{\ENDEITHER}{\end{ALC@either} \STATE \textbf{end either}}

\newcommand{\SWITCH}[1]{\STATE \textbf{switch} #1 \begin{ALC@g}}
\newcommand{\ENDSWITCH}{\end{ALC@g} \STATE \textbf{end switch}}
\newcommand{\CASE}[1]{\STATE \textbf{case} #1\textbf{:} \begin{ALC@g}}
\newcommand{\ENDCASE}{\end{ALC@g}}

\newcommand{\DEFAULT}{\STATE \textbf{default:} \begin{ALC@g}}
\newcommand{\ENDDEFAULT}{\end{ALC@g}}

  \newcommand{\ELSIFA}[2][default]%
{\end{ALC@if} \STATE d \begin{ALC@if}}

\begin{document}
\title[The Algebraic Intersection Type Unification Problem]{The Algebraic Intersection Type Unification Problem}

\author[A.~Dudenhefner]{Andrej Dudenhefner}	
\address{Department of Computer Science, Technical University of Dortmund, Dortmund, Germany}	
\email{\{andrej.dudenhefner,moritz.martens,jakob.rehof\}@cs.tu-dortmund.de}  

\author[M.~Martens]{Moritz Martens}	

\author[J.~Rehof]{Jakob Rehof}	




\begin{abstract}
The algebraic intersection type unification problem is an important component in
proof search related to several natural decision problems in intersection type systems.
It is unknown and remains open whether the algebraic intersection type unification problem is decidable. 
We give the first nontrivial lower bound for the problem by showing (our main result) that it
is exponential time hard. Furthermore, we show that this holds even under rank 1 solutions (substitutions whose codomains are restricted to contain rank 1 types). In addition,
we provide a fixed-parameter intractability result for intersection type matching (one-sided unification),
which is known to be NP-complete.

We place the algebraic intersection type unification problem in the context of unification theory.
The equational theory of intersection types
can be presented as an algebraic theory with an ACI (associative, commutative, and idempotent) operator (intersection type) combined with distributivity properties with respect to a second operator (function type).
Although the problem is algebraically natural and interesting, it appears to
occupy a hitherto unstudied place in the theory of unification, and our investigation of the problem 
suggests that new methods are required to understand the problem. Thus,
for the lower bound proof, we were not able to reduce from known results in ACI-unification theory and
use game-theoretic methods for two-player tiling games.
 \end{abstract}

\maketitle

\section{Introduction}

Intersection type systems occupy a prominent place within the theory of typed $\lambda$-calculus \cite{BDS13}.
As is well known, variants of such systems characterize deep semantic properties of $\lambda$-terms, including normalization and solvability properties \cite{BDS13}. As a consequence of the enormous expressive power of intersection types, standard type-theoretic decision problems are undecidable for general intersection type systems, including the problem of type checking (given a term and a type, does the term have the type?) and inhabitation (given a type, does there exist a term having the type?). A combinatorial problem centrally placed in many classical type-theoretic decision problems is that of {\em 
unification} over a term algebra of type expressions: given two types $\sigma$ and $\tau$, does there exist a substitution $S$ of types for type variables such that $S(\sigma) = S(\tau)$ in a suitable equational theory ($=$) of types? In this paper we wish to initiate a study of the problem of {\em algebraic  intersection type unification} which we believe to be of considerable systematic interest. 
We consider an equational theory of intersection types induced by an important subtyping relation for intersection types \cite{bcd}
extended with type constants. Although decidability of algebraic intersection type unification appears to be surprisingly difficult and remains open, the present paper provides the first nontrivial lower bound indicating that the problem is of very high complexity: we prove that the problem is {\sc Exptime}-hard. Our proof uses game-theoretic methods, in the form of two-player tiling games, which we believe to be of intrinsic interest and potentially helpful towards understanding the problem of decidability. Moreover, as we will show, the algebraic intersection type unification problem occupies a natural but hitherto (so far as we are aware) unstudied place in the theory of unification. 

The algebraic unification problem considered in this paper is distinct from problems of unification with type schemes under {\em chains of
substitutions and expansions}, as found in, e.g., the work by Ronchi Della Rocca \cite{RonchiDellaRocca88}. This line of work was motivated
from the theory of principality for intersection types but contains the algebraic intersection type unification problem as a component.
In the following, unless otherwise stated, when we talk about ``unification" we shall understand algebraic unification considered here. 

We hope with this paper to stimulate further work on a fascinating open problem in type theory as well
as in unification theory. Many variants of intersection type subtyping have been introduced (see \cite{BDS13} for an overview). 
For all of these theories corresponding algebraic intersection type unification problems can be induced along the lines of this paper, and consequently we should think of the algebraic unification problem considered here for one specific theory to be indicative of a whole {\em family} of such problems. For many, probably most, of these variations it is not known whether algebraic unification is decidable. The present paper establishes the result that for one important variant the problem is of high complexity and suggests that the question of decidability is both challenging and interesting for intersection type unification in general.

\subsection{Motivation}
\label{subsec:motivation}
Since the algebraic unification problem under study in this paper has not hitherto been considered {\em per se},
we discuss some general motivations for considering the problem here, whereas more specific remarks can be found in Sec.~\ref{sec:related}.

First, as will be explained in more detail below and in Sec.~\ref{subsec:related-unification}, the algebraic unification problem for intersection types
appears be a natural problem of systematic interest in unification theory \cite{Baader01}, occupying a hitherto unstudied position there and posing
interesting challenges which do not seem to be readily amenable to known techniques.

The second main motivation for studying the problem comes from applications of the theory of intersection types in programming.
It is well known that important decision problems (most prominently, typability and type checking) 
for various type systems can be characterized, modulo polynomial time transformations, by certain algebraic and combinatorial problems,
which do not refer to the rules of the type system at all, thereby providing abstract combinatorial and algorithmic ``signatures" for these systems.
To mention a few well known examples, typability in simple types can be characterized as a standard unification problem over the free term algebra of simple type expressions 
and can be solved using the unification algorithm of Robinson \cite[Sec. 3D]{Hindley08}. Typability in systems extending simple types with
subtyping can be characterized by the problem of subtype satisfiability \cite{JR98}. Typability in ML can be characterized by acyclic
semi-unification problems, whereas the problem becomes equivalent to general semi-unification in the presence of
polymorphic recursion and is undecidable \cite{KTU93}. In contrast to these examples, it does not appear to be obvious, in general, how to characterize typability problems for intersection type systems in terms of such algebraic type constraint systems. The reason is, essentially, that intersection type systems are ``proof functional" (rather than ``truth functional"). More precisely, the rule of intersection introduction requires {\em one and the same term} to have both
types $\tau$ and $\sigma$ for the term to have the type $\tau \cap \sigma$. Because of this dependency on the identity of the subject term, it is not clear how the logical effects of the rule can be captured in a purely algebraic way. If we would attempt to characterize typability by extracting algebraic type constraints from a given term, we would have difficulty deciding, from the structure of term alone, how to apply the rule of intersection introduction. 

Nonetheless, it turns out that intersection type unification as considered here is a natural {\em component} in deciding problems related to typability 
and type checking for intersection types. As already mentioned, the algebraic unification problem considered in this paper is a component problem in the problem of unification with type schemes under {\em chains of
substitutions and expansions} studied by Ronchi Della Rocca \cite{RonchiDellaRocca88} 
(assuming we consider the problem with usual type constants, as suggested
in \cite[p.182]{RonchiDellaRocca88}), for which a semi-decision procedure is provided. Further remarks comparing to this notion of unification can be found in Sec.~\ref{subsec:related-type-theory}.
The algebraic unification problem is relevant when investigating restrictions of the intersection type system. Such investigation, in turn, may be of particular interest, because the full system is undecidable (both with respect to typability, type checking, and inhabitation, see \cite{BDS13} for an overview).
An example of such a restriction is provided in Sec~\ref{subsec:unification} (Example~\ref{ex:typability}), where we consider combinators with intersection types over arbitrary bases but excluding the intersction introduction rule, motivated from applications in program synthesis. 
The typability problem for this fragment is polynomial time equivalent to intersection type unification.
Another example is the restriction of $\lambda$-calculus with explicit intersection types studied in \cite{KT95}, which is equivalent
to the intersection type system of \cite{bcd} without intersection introduction.
We know that type checking is decidable for $\beta$-normal forms in this system, as was shown in \cite{KT95}, and inhabitation is decidable and exponential space complete
\cite{RehofU12}. But we do not know whether type checking is decidable in general for this restriction. However, it is not difficult to see that the problem is reducible to intersection type unification. An interesting question to consider here is whether it might be equivalent to either the unification problem or to its associated {\em matching problem} considered in this paper. 
The algebraic unification problem is likely to be involved as soon as one attempts to combine intersection types with usual notions of type instantiation. A recent example is the 
so-called type tallying problem of \cite{CNXA15}, which is decidable. The main difference to the problem considered in this paper is the inclusion of union and negation type connectives as well as recursive types. Recursive types in particular allow for a direct solution of cyclic constraints such as $\alpha \doteq \alpha \cap (\alpha \to b)$ by $S(\alpha) = \mu X.(X \cap (X \to b))$. For all we know at this time, it could be the case that the algebraic intersection type unification problem considered in this paper is undecidable, but that it becomes decidable when recursive types are added.
Finally, let us mention that it may be interesting, in the light of this paper, to compare the combinatorial complexity of the theory of subtyping and type equivalence in the {\em strict} intersection type system \cite{vBakel2011} with that of, say, \cite{bcd}, by considering 
associated algebraic unification problems.
\newpage\noindent
Perhaps these theories will turn out to be distinguishable in that respect, but the question of decidability of unification is open for both systems.

It should be emphasized that, as already indicated, many variants of intersection type subtyping theories exist (a catalogue can be found in \cite{BDS13}), 
giving rise to a whole {\em family} of intersection type unification problems analogous to the one studied in this paper.
While we discuss some variations, other members of this family of problems are yet to be studied. Our choice of the specific variant of
subtyping and type equivalence arises from the theory presented in
\cite{bcd} extended with type constants. This choice was motivated by two main considerations. For one, the theory of subtyping and equivalence of \cite{bcd} is particularly rich,
and several variants can be seen as subsystems thereof. Second, we adopted this theory for applications within type-based
program synthesis based on combinatory logic with intersection types \cite{JR13,DMResop14}, where it is useful for supporting type refinement
with semantic types (compare also \cite{frepfe91}). Similar theories of intersection type subtyping and equivalence have been adopted, or adapted, in many contexts, an example being the programming language Forsythe by Reynolds \cite{Forsythe}. It should be noted that the exact theory of 
intersection types considered in \cite{bcd} contains only a single type constant, denoted 
$\omega$, which is equivalent to the empty intersection and is semantically a universal type. 
Because of the special properties of the type $\omega$ (in particular, the equation $\tau \to \omega = \omega$), the intersection type unification problem is trivial for this exact variant, because we can solve any such problem by mapping all type variables to $\omega$ (see also further remarks in Sec.~\ref{sec:related}). However, if we leave out the type $\omega$ (as is done in some variations), or if we introduce other type constants 
(such as $\texttt{bool}$, $\texttt{int}$, $\texttt{real}$), as is done in applications to programming languages, we arrive at theories for which
the development in this paper is relevant. A more detailed analysis of further variants must, however, be postponed for future work
to which the present paper may be seen as an invitation.

\subsection{Algebraic properties}
\label{subsec:algebraic-properties}

We briefly summarize some of the most important algebraic properties of the equational theory of intersection types needed to appreciate the systematic placement of the unification problem (full details are given later in the paper). Intersection type systems are characterized by the presence of an associative, commutative, idempotent operator, $\cap$ (intersection), which allows the formation of types of the form $\sigma \cap \tau$. In addition, we have function types, $\sigma \to \tau$. The standard equational theory, denoted $=$, of intersection types \cite{bcd} is induced from a partial order $\le$ on types, referred to as subtyping, 
by taking type equality to be the relation $\le \cap \le^{-1}$. Conversely, as will be discussed in the paper, it is also possible to give a purely equational presentation of subtyping. Because intersection is  greatest lower bound with respect to subtyping, the intersection type unification problem is equivalent to the subtype satisfiability problem: given $\sigma$ and $\tau$, does there exist a type substitution $S$ such that $S(\sigma) \le S(\tau)$? The latter is equivalent to $S(\sigma) \cap S(\tau) = S(\tau)$, hence satisfiability is reducible to unification. The equational theory includes right-distributivity of $\to$ over $\cap$: $\sigma \to (\tau_1 \cap \tau_2) = (\sigma \to \tau_1) \cap (\sigma \to \tau_2)$ and left-contravariance of $\to$ with respect to subtyping: $\sigma_1 \to \tau_1 \le \sigma_2 \to \tau_2$ whenever $\sigma_2 \le \sigma_1$ and $\tau_1 \le \tau_2$. As a consequence, one has ``half left-distributivity'' of $\to$ over $\cap$: $(\sigma_1 \to \tau) \cap (\sigma_2 \to \tau) \le (\sigma_1 \cap \sigma_2) \to \tau$ (but the symmetric relation does not hold). Altogether, we could say for short that $\to$ is ``$1\frac{1}{2}$-distributive'' over $\cap$. Axioms specific to a special largest type, $\omega$, are added in some variants of the theory (both variants, with or without $\omega$, are important in type theory), including the recursion axiom $\omega = \omega \to \omega$, and we have the derived equation $\sigma \to \omega = \omega$.
Thus, $\omega$ is unit (neutral element) 
with respect to $\cap$ and right-absorbing element with respect to $\to$. Clearly, as already mentioned, in a system in which $\omega$ is the only constant any two intersection types are unifiable by replacing all occurring variables by $\omega$. The variant of the problem considered in the following includes
other type constants besides $\omega$ as well.

\section{Related work}
\label{sec:related}

We consider the most closely related work within unification theory and type theory.

\subsection{Related work in unification theory}
\label{subsec:related-unification}
The single most directly related piece of work in the literature is the study from $2004$ by
Anantharaman, Narendran, and Rusinowitch on unification modulo ACUI (associativity, commutativity,
unit, idempotence) plus distributivity axioms 
\cite{anantharaman2004}. They consider equational theories over a binary ACUI symbol, denoted $+$, together with a binary operator, $*$, which distributes (left, right, or both) over $+$. Indeed, since (as summarized above) we have an ACUI theory  of $\cap$ together with $\to$ enjoying distributivity properties over $\cap$, it would seem that we are temptingly close to the theories
studied in \cite{anantharaman2004}, by thinking of their $+$ as $\cap$ and their $*$ as $\to$. In particular, algebraically closest among the theories covered in that paper, ACUI-unification with one-sided (say, left) distributivity ($\mbox{ACUID}_l$) is shown to be {\sc Exptime}-complete, using techniques from unification modulo homomorphisms \cite{narendran2001}. But it turns out that 
there are fundamental obstacles to transferring results or techniques from 
$\mbox{ACUID}_l$-unification to intersection type unification, as will be summarized next. 

With regard to any upper bound, the main obstacle is that, whereas decidability of
the ACUID-problems can be relatively straight-forwardly obtained by appeal to an occurs-check (nontrivial cyclic equations have no solutions), this is very far from being clear in the case of 
intersection type unification.
Indeed, even in the absence of the recursive type $\omega$, we can solve nontrivial cyclic constraints, due to contravariance. For example, the constraint $\alpha \dotleq \alpha \to b$ (where $\dot{\leq}$ 
denotes a formal subtyping constraint, $\alpha$ is a type variable and 
$b$ is a constant) can be solved by setting $S(\alpha) = b \cap (b \to b)$. The theory of intersection types is {\em non-structural} in the sense that types  with significantly different shapes (tree domains, when types are regarded as labeled trees)  may be related, and this presents fundamental obstacles for bounding the depth of substitutions via any kind of standard occurs-check. Although \enquote{$1\frac{1}{2}$-distributivity} of $\to$ over $\cap$ may at first sight appear to be algebraically close to the ACUID-framework of \cite{anantharaman2004}, the contravariant 
\enquote{$\frac{1}{2}$-distributivity} makes the theory of intersection types significantly different.  We cannot exclude that some kind of restricted occurs-check might be possible, but our investigations lead us to believe that, in case it exists, it is likely to be very complicated, and we have been unable to find such a bounding principle. Hence, decidability remains a challenging open problem.

With regard to the exponential time lower bound, the results of \cite{anantharaman2004} (in fact, both the {\sc Exptime}
upper and lower bounds) rely essentially on 
reductions from unification modulo a set $H$ of noncommuting homomorphisms
(ACUIDH), which was shown to be {\sc Exptime}-complete in \cite{narendran2001}. The basic idea is to represent unification with distributivity to unification modulo homomorphisms by replacing $s * t$ by $h_s(t)$ where $h_s$ is a homomorphism with respect to the AC(U)I-theory. However, again, such techniques fail in our case due to contravariance. The equational presentation of the theory of intersection types captures contravariant subtyping by the absorption axiom (written in the algebraic notation of \cite{anantharaman2004}): $s * t = s * t + (s + s') * t$. One could attempt to represent this axiom by 
$h_s(t) = h_s(t) + h_{s+s'}(t)$. But here the expression $h_{s+s'}(t)$ does not fall within the homomorphic format, and it is therefore not clear how the homomorphic framework could be applied. Moreover, the bounding problem discussed above leads to the problem that it is not clear how the theory could be adequately represented using only a finite set of homomorphisms.
We concluded that new methods are required in order to provide lower bounds for intersection type unification, and the route we present in this paper for the {\sc Exptime}-lower bound is entirely different, relying on game theoretical results on tiling problems.

\subsection{Related work in type theory}
\label{subsec:related-type-theory}
It may be surprising that computational properties (decidability,
complexity) of the intersection type unification problem have not previously been 
systematically pursued {\em per se}. 
The theory of intersection type subtyping and its equational counterpart have rather been studied
from semantic (operational and denotational) perspectives. Indeed, as mentioned already, the intersection type system captures deep operational properties of $\lambda$-terms, and undecidability of type checking and typability follows immediately. The theories
of intersection type subtyping and equality studied here arose naturally out of model-theoretic considerations. 
For example, a fundamental result \cite{Hindley82,bcd} shows that intersection type subtyping and equality are sound and complete for set-theoretic containment in a class of 
$\lambda$-models: $\sigma \le \tau$ holds, if and only if
$\llbracket \sigma \rrbracket^{\mathcal{M}}_v \subseteq \llbracket \tau \rrbracket^{\mathcal{M}}_v$
for all models $\mathcal{M}$ in the class and valuations $v$. The intersection type unification problem can therefore also be endowed with semantic interpretations.


Several extensions and variations of the standard algebraic operations of unification studied here
have been considered in connection with intersection type systems, foremostly motivated by 
questions related to notions of principality (principal types, principal typings, principal pairs) in such systems. 
Ronchi della Rocca, working from such motivations, defines a notion of unification in \cite{RonchiDellaRocca88}
and gives a semi-decision procedure for the corresponding unification problem. But that problem involves operations (chains of substitutions together with special expansion operations) which are not present in the algebraic notion of unification we consider here. Similarly, so-called expansion variables with associated operations have been used by Kfoury and Wells to characterize principality properties \cite{Wells03}
and so-called $\beta$-unification involving expansion variables has been shown to characterize strong normalization in the $\lambda$-calculus \cite{Kfoury99}, see also \cite{Coppo95,Boudol05}. 

Since the contravariant \enquote{$\frac{1}{2}$-distributivity} makes the theory of intersection types significantly different compared to the well studied $\mbox{ACUID}_l$-unification, it is interesting to compare with problems that include left-contravariance of $\to$. One of them is the subtyping problem for second-order types, which is undecidable~\cite{tiuryn1996subtyping}. The authors point out that, if $\to$ is considered covariant in both arguments, then the corresponding subtyping problem becomes decidable. Finally, the above mentioned non-structural nature of the unification problem
and the challenges it poses for bounding substitution depth could possibly be compared with  
challenges posed by the well known open problem of {\em non-structural subtype entailment}. For this problem
a \textsc{Pspace} lower bound is known \cite{HengleinRehof98}, but despite many attempts over the past $20$ years the question of decidability
of non-structural subtype entailment remains open (see \cite{TLCAlist} for more information and further references).

Summarizing the situation with regard to intersection type unification within type theory, it appears
to hold an interesting and rather unexplored intermediate position: it is contained in many decision problems associated with intersection type systems, it is known to be expressive enough to capture certain restrictions of the type system, but it is not known whether it is decidable. It is therefore also a problem of importance for advancing our understanding of restrictions of the intersection type system and computational properties of associated decision problems.

{\bf Organization of the paper}. The remainder of this paper is organized as follows.
Intersection types are introduced in Sec.~\ref{sec:intersection-types} together with the standard theory of subtyping \cite{bcd}. In Sec.~\ref{sec:matching} we briefly study the matching problem (one-sided unification) as a natural preparation for considering the unification problem. The unification problem is studied in 
Sec.~\ref{sec:unification}, which contains our main result.
We first introduce the unification problem and the equational theory of intersection types
(Sec.~\ref{subsec:unification}) and then turn to the proof of the \textsc{Exptime}-lower bound. We introduce 
tiling games (Sec.~\ref{sec:tiling-games}) and prove \textsc{Exptime}-completeness of a special form of such (``spiral tiling games''), which is then used (Sec.~\ref{sec:exptime-lb}) in our reduction to unification and satisfiability. Next we adapt the \textsc{Exptime}-lower bound to the $\omega$-free fragment (Sec.~\ref{sec:exptime-lb-no-omega}) and inspect a different presentation of the underlying equational theory (Sec.~\ref{sec:alternative-presentation}). We provide some insights into a possible upper bound construction for unification in rank 1 (Sec.~\ref{sec:rank1}) and conclude the paper in Sec.~\ref{sec:conclusion}.

\section{Intersection types}
\label{sec:intersection-types}
\begin{definition}[$\TT$]
The set $\TT$ of intersection types, ranged over by $\sigma, \tau, \rho$, is given by
$$\TT \ni \sigma, \tau, \rho \ ::= \  a \mid \alpha \mid \omega \mid \sigma \rightarrow \tau \mid \sigma \cap \tau$$
where $a, b, c, \ldots$ range over type constants drawn from the set $\CC$, $\omega$ is a special (universal type) constant, and $\alpha, \beta, \gamma$ range over type variables drawn from the set $\VV$.
\end{definition}
As a matter of notational convention, function types associate to the right, and $\cap$ binds stronger than $\to$. A type $\tau \cap \sigma$ is said to have $\tau$ and $\sigma$ as {\em components}. The {\em size} of a type $\tau$, denoted by $|\tau|$, is the number of nodes in the syntax tree of $\tau$.

\begin{definition}[Subtyping $\leq$]
\label{def:subtyping}
Subtyping $\leq$ is the least preorder (reflexive and transitive relation) over $\TT$ (cf.~\cite{bcd}) such that
\begin{align*}
&\sigma \leq \omega, \quad \omega \leq \omega\to\omega, \quad \sigma \cap \tau \leq \sigma, \quad \sigma \cap \tau \leq \tau, \quad (\sigma\to\tau_1)  \cap  (\sigma\to\tau_2) \leq \sigma \to \tau_1 \cap \tau_2,\\
&\text{if }\sigma \leq \tau_1 \text{ and } \sigma\leq \tau_2 \text{ then } \sigma \leq \tau_1 \cap \tau_2, \quad
\text{if } \sigma_2 \leq \sigma_1 \text{ and } \tau_1 \leq \tau_2 \text{ then } \sigma_1\to\tau_1 \leq \sigma_2\to\tau_2
\end{align*}
\end{definition}
Type equality, written $\sigma = \tau$, holds when $\sigma\leq\tau$ and $\tau\leq\sigma$, thereby
making $\leq$ a partial order over $\TT$. We use $\equiv$ for syntactic identity.
By the axioms of subtyping, $\cap$ is associative, commutative, idempotent and has the following distributivity properties
$$(\sigma \to \tau_1) \cap (\sigma \to \tau_2) = \sigma \to (\tau_1 \cap \tau_2) 
\qquad
(\sigma_1 \to \tau_1) \cap (\sigma_2 \to \tau_2) \leq (\sigma_1 \cap \sigma_2) 
\to (\tau_1 \cap \tau_2) $$
We write $\bigcap_{i = 1}^n \tau_i $ or 
$\bigcap_{i \in I} \tau_i $ or $\bigcap\{ \tau_i \mid i \in I \}$ for an intersection of several components,  
where the empty intersection is identified with~$\omega$. 

Observe that in the above subtyping definition type constants are treated the same as type variables. The motivation to distinguish type variables from type constants is twofold.
First, practical scenarios often contain constants (e.g. \texttt{int} or \texttt{string}) that cannot be instantiated.
Second, in the setting of E-unification in the sense of~\cite{Baader01}, variables that are subject to instantiation by terms in the underlying term algebra are distinct from any constructors of that algebra (in particular variables in the sense of~\cite{bcd}).

Using \cite{bcd}(Lemma 2.4.1) we syntactically define the set $\TT^\omega$ of types equal to $\omega$.

\begin{definition}[$\TT^\omega$] The set $\TT^\omega$ of types in $\TT$ equal to $\omega$ is given by
$$ \TT^\omega \ni \sigma^\omega, \tau^\omega ::= \omega \mid \sigma \to \tau^\omega \mid \sigma^\omega \cap \tau^\omega $$
\end{definition}

\begin{lemma}
\label{lem:tt_omega}
For $\tau \in \TT$ we have $\tau \in \TT^\omega$ iff $\tau = \omega$.
\end{lemma}

\begin{lemma}[Beta-Soundness \cite{bcd,BDS13}]
\label{lem:beta_soundness}
Given $\sigma = \bigcap\limits_{i\in I}(\sigma_i\to\tau_i) \cap \bigcap\limits_{j\in J}a_j \cap \bigcap\limits_{k\in K}\alpha_k$, we have
\begin{enumerate}[label=(\roman*)]
\item If $\sigma \leq a$ for some $a \in \CC$, then $a \equiv a_j$ for some $j \in J$.
\item If $\sigma \leq \alpha$ for some $\alpha \in \VV$, then $\alpha \equiv \alpha_k$ for some $k \in K$.

\item If
$\sigma \leq \sigma' \to \tau' \neq \omega$ for some $\sigma', \tau' \in \TT$,
then $I' = \{i \in I \mid \sigma' \leq\sigma_i\} \neq \emptyset$ and
$\bigcap\limits_{i \in I'} \tau_i \leq \tau'$.
\end{enumerate}
\end{lemma}

Let us briefly consider the problem of deciding the subtyping relation itself:
\begin{problem}(Subtyping)
\label{prob:subtyping}
Given $\sigma, \tau \in \TT$, does $\sigma \leq \tau$ hold?
\end{problem}

The subtyping relation is known to be decidable in polynomial time \cite{RehofUrzyczyn11}. 
The algorithm sketched in the proof of the following
lemma gives an improved, quadratic upper bound (implemented in the Combinatory Logic Synthesis framework \cite{BDDMR14} which is under continuous development). 
\begin{lemma}
Problem \ref{prob:subtyping} (Subtyping) is decidable in
time $\calO(n^2)$ where $n$ is the sum of the sizes of the input types $\sigma$ and $\tau$.
\end{lemma}

\begin{proof}
For a polynomial time decision algorithm with a quartic upper bound see~\cite{RehofUrzyczyn11}. For a different approach with a quintic upper bound using rewriting see~\cite{statman2015}.
However, a quadratic upper bound to decide $\sigma \leq \tau$ is achievable using Lemmas \ref{lem:tt_omega} and \ref{lem:beta_soundness}. First, in linear time, subterms of $\sigma$ and $\tau$ of the shape defined by $\TT^\omega$ are replaced by $\omega$. Second, in linear time, nested intersection are flattened using associativity of $\cap$ and components equal to $\omega$ are dropped. Third, in quadratic time, Lemma \ref{lem:beta_soundness} is applied recursively using the additional property $\rho \leq \bigcap_{i \in I} \tau_i$ iff $\rho \leq \tau_i$ for $i \in I$.

For the explicit algorithmic analysis let $T(\sigma, \tau)$ be the running time to decide $\sigma \leq \tau$. We show $T(\sigma, \tau) \leq a |\sigma| |\tau| - b = \calO(|\sigma||\tau|)$ for some constants $a,b \in \mathbb{R}_+$ with $a \geq b$ by induction on the sum of the depths of syntax trees of $\sigma$ and $\tau$.
The most challenging subcase in the algorithmic application of Lemma \ref{lem:beta_soundness} is to decide whether $\bigcap_{i = 1}^n \sigma_i \leq \sigma' \to \tau'$ holds. We proceed as follows: let $S$ be an empty linked list. For $i = 1 \ldots n$ such that $\sigma_i \equiv \rho_i \to \tau_i$ and $\sigma' \leq \rho_i$ append $\tau_i$ to $S$. Finally, accept iff $\bigcap S \leq \tau'$. Our goal is to bound $T(\sigma, \sigma' \to \tau')$ where $\sigma \equiv \bigcap_{i = 1}^n \sigma_i$ by $a |\sigma| |\sigma' \to \tau'| - b$. Let $\tau \equiv \bigcap(\tau_i \mid \sigma_i \equiv \rho_i \to \tau_i, \sigma' \leq \rho_i)$
observing that the depth of the syntax tree of $\tau$ is strictly less than the depth of the syntax tree of $\sigma$.
By the induction hypothesis for some constant $u$ we have
\begin{align*}
T(\sigma, \sigma' \to \tau') &\leq u n + \sum\limits_{\substack{i \in \{1, \ldots, n\} \\ \sigma_i \equiv \rho_i \to \tau_i}} T(\sigma',\rho_i) + T(\tau, \tau')\\
&\leq u n + \sum\limits_{\substack{i \in \{1, \ldots, n\} \\ \sigma_i \equiv \rho_i \to \tau_i}} (a |\sigma'| |\rho_i| - b) + a |\tau| |\tau'| - b\\
&\leq u n + a |\sigma'| \sum\limits_{\substack{i \in \{1, \ldots, n\} \\ \sigma_i \equiv \rho_i \to \tau_i}} |\rho_i| + a |\tau| |\tau'| - b\\
&\leq u n + a |\sigma'| |\sigma| + a |\sigma| |\tau'| - b\\
&\leq u n + a |\sigma| (|\sigma' \to \tau'| - 1) - b
\stackrel{a \geq u, |\sigma| \geq n}{\leq} a |\sigma| |\sigma' \to \tau'| - b
\end{align*}
In the above analysis, it is crucial that linked list concatenation is done in constant time in the construction of $\tau$. Therefore, the invariant that $\cap$ is not nested and does not contain $\omega$ as components can be ensured in recursive calls. 
\end{proof}

We recapitulate the notion of paths and organized types introduced in~\cite{RehofEtAlCSL12} that coincides with the notion of factors in~\cite{statman2015}.

\begin{definition}[Paths $\PP$] 
\label{def:path}
The set $\PP$ of paths in $\TT$, ranged over by $\pi$, is given by
$$ \PP \ni \pi ::= a \mid \alpha \mid \tau \to \pi $$
\end{definition}

\begin{definition}[Organized type]
A type $\tau$ is \emph{organized}, if $\tau \equiv \omega$ or $\tau \equiv \bigcap_{i \in I} \pi_i$ for some paths $\pi_i$ for $i \in I$. 
\end{definition}
A type can be organized (transformed to an equivalent organized type) in polynomial time~\cite{RehofEtAlCSL12}. Additionally, an organized type is not necessarily normalized
in the sense of~\cite{Hindley82}. Normalization can lead to an exponential blow-up of type size, whereas organization does not as illustrated by the following examples.

\begin{example}
Let $\sigma \equiv ((a \cap b \to a \cap b) \to a \cap b) \to a \cap b$. The equivalent organized type is $\sigma = \big(((a \cap b \to a \cap b) \to a \cap b) \to a \big)  \cap \big(((a \cap b \to a \cap b) \to a \cap b) \to b \big)$, whereas the equivalent normalized type in the sense of~\cite{Hindley82} is
\begin{align*}
\sigma =& \big(((a \cap b \to a) \cap (a \cap b \to b) \to a) \cap ((a \cap b \to a) \cap (a \cap b \to b) \to b) \to a \big)\\
\cap & \big(((a \cap b \to a) \cap (a \cap b \to b) \to a) \cap ((a \cap b \to a) \cap (a \cap b \to b) \to b) \to b \big)
\end{align*}
In general, normalization can lead to an exponential blow-up of type size because arguments of arrow types are normalized recursively.
\end{example}

\begin{example}
Let $\sigma \equiv a \to \bigcap\limits_{i=0}^1 (b_i \to \bigcap\limits_{j=0}^1
(c_j \to \bigcap\limits_{k=0}^1 d_k))$. The equivalent organized type coincides
with the equivalent normalized type $\sigma = \
\mathclap{\bigcap\limits_{(i,j,k) \in \{0,1\}^3}}\ \  (a \to b_i \to c_j \to d_k)$.
Still, when generalized, this example does not lead to an exponential blow-up of type size by organization because the size of  $\sigma$ is already exponential in the depth of its syntax tree.
\end{example}

\begin{lemma}\label{lem:organized_subtyping}
Given two organized types $\sigma \equiv \bigcap_{i \in I} \pi_i$ and $\tau \equiv \bigcap_{j \in J} \pi_j$, we have \\
$\sigma \leq \tau$ iff for all $j \in J$ there exists an $i \in I$ with $\pi_i \leq \pi_j$.
\end{lemma}

\begin{corollary}
\label{cor:path_subtyping}
Given a path $\pi \in \PP$ and types $\sigma, \tau$, we have 
$\sigma \cap \tau \leq \pi$ iff $\sigma \leq \pi$ or $\tau \leq \pi$.
\end{corollary}

For the sake of completeness, we recall the corresponding type assignment system from~\cite{bcd}, sometimes also called {\bf BCD} in literature. A basis (also called context) is a finite set $\Gamma = \{x_1 : \tau_1, \ldots, x_n : \tau_n\}$, where the variables $x_i$ are pairwise distinct; we set $\dom(\Gamma) = \{x_1, \ldots, x_n\}$ and we write $\Gamma, x : \tau$ for $\Gamma \cup \{x : \tau\}$, where $x \not\in \dom(\Gamma)$.

\begin{definition}[Type Assignment] 
\label{def:bcd}
{\bf BCD} type assignment is given by the following rules
\medskip

\begin{tabular}{c c c}
{\RightLabel{\AX}
\AxiomC{$x : \tau \in \Gamma$}
\UnaryInfC{$\Gamma \vdash x : \tau$}
\DisplayProof}
& {\RightLabel{\AI}
\AxiomC{$\Gamma, x : \sigma \vdash e : \tau$}
\UnaryInfC{$\Gamma \vdash \lambda x.e : \sigma \to \tau$}
\DisplayProof} 
& {\RightLabel{\AE}
\AxiomC{$\Gamma \vdash e : \sigma \to \tau$}
\AxiomC{$\Gamma \vdash e' : \sigma$}
\BinaryInfC{$\Gamma \vdash (e\ e') : \tau$}
\DisplayProof}\\
\\
{\RightLabel{\OM}
\AxiomC{}
\UnaryInfC{$\Gamma \vdash e : \omega$}
\DisplayProof}
& {\RightLabel{\II}
\AxiomC{$\Gamma \vdash e : \sigma$}
\AxiomC{$\Gamma \vdash e : \tau$}
\BinaryInfC{$\Gamma \vdash e : \sigma \cap \tau$}
\DisplayProof} 
& {\RightLabel{\SUB}
\AxiomC{$\Gamma \vdash e : \sigma$}
\AxiomC{$\sigma \leq \tau$}
\BinaryInfC{$\Gamma \vdash e : \tau$}
\DisplayProof}
\end{tabular}
\end{definition}

\section{Intersection type matching}
\label{sec:matching}

In order to understand the unification problem it is useful first to investigate its restriction to matching (one-sided unification). Intersection type matching occurs naturally during proof search in intersection type systems and is known to be \textsc{NP}-complete \cite{RehofEtAlTLCA13}. We strengthen this result by showing that the problem remains so even when restricted to the fixed-parameter case where 
only a single type variable and only a single constant is used in the input.

For $\tau \in \TT$ let $\Var(\tau) \subseteq \VV$ denote the set of variables occurring in $\tau$.

\begin{problem}[Matching]
\label{prob:matching}
Given a set of constraints $C = \{\sigma_1 \dotleq \tau_1, \ldots, \sigma_n \dotleq \tau_n\}$, where for each $i \in \{1, \ldots, n\}$ we have $\Var(\sigma_i) = \emptyset$ or $\Var(\tau_i) = \emptyset$,
is there a substitution $S \colon \VV \to \TT$ such that $S(\sigma_i) \leq S(\tau_i)$ for $1 \leq i \leq n$?
\end{problem}
We say that a substitution $S$ satisfies $\{\sigma_1 \dotleq \tau_1, \ldots, \sigma_n \dotleq \tau_n\}$ if $S(\sigma_i) \leq S(\tau_i)$ for $1 \leq i \leq n$.

\begin{problem}[One-Sided Unification]
\label{prob:one-sided_unification}
Given a set of constraints $C = \{\sigma_1 \doteq \tau_1, \ldots, \sigma_n \doteq \tau_n\}$, where for each $i \in \{1, \ldots, n\}$ we have $\Var(\sigma_i) = \emptyset$ or $\Var(\tau_i) = \emptyset$,
is there a substitution $S \colon \VV \to \TT$ such that $S(\sigma_i) = S(\tau_i)$ for $1 \leq i \leq n$?
\end{problem}

Any matching constraint set $C = \{\sigma_1 \dotleq \tau_1, \ldots, \sigma_n \dotleq \tau_n\}$ can be reduced to a single matching (resp. one-sided unification) constraint $\sigma \dotleq \tau$ (resp. $\sigma \cap \tau \doteq \sigma$) with $\Var(\sigma) = \emptyset$ by fixing a type constant $\bullet \in \CC$, defining 
$$ (\sigma_i', \tau_i') = \begin{cases}
(\sigma_i \to \bullet, \tau_i \to \bullet) & \text{if } \Var(\sigma_i) = \emptyset\\
(\tau_i, \sigma_i) & \text{if } \Var(\tau_i) = \emptyset\\
\end{cases} \quad \text{for } 1 \leq i \leq n$$
and taking $\sigma \equiv \sigma_1' \to \ldots \to \sigma_n' \to \bullet$ and $\tau \equiv \tau_1' \to \ldots \to \tau_n' \to \bullet$.
By Lemma \ref{lem:beta_soundness}, for any substitution $S$ we have $S(\sigma) \leq S(\tau)$ (resp. $S(\sigma \cap \tau) = S(\sigma)$) iff $S(\sigma_i) \leq S(\tau_i)$ for $1 \leq i \leq n$. Observe that $\Var(\sigma) = \emptyset$. Therefore, matching and one-sided unification coincide and remain \textsc{NP}-complete even when restricted to single constraints.

In~\cite{RehofEtAlTLCA13} the lower bound for matching is shown by reduction from 3-SAT and requires two type variables $\alpha_x, \alpha_{\neg x}$ for each propositional variable $x$. Since 3-SAT, parameterized by the number of propositional variables, is fixed parameter tractable, we naturally ask whether the same holds for matching (resp. one-sided unification) parameterized by the number of type variables. 

\begin{lemma}
\label{lem:matching_one_var}
Problem \ref{prob:matching} (Matching) is \textsc{NP}-hard even if only a single type variable and a single constant is used in the input.
\end{lemma}

\begin{proof}
We fix a 3-SAT instance $F$ containing clauses $(L_1 \vee L_2 \vee L_3) \in F$ over propositional variables $V$ where $L_i$ is either $x$ or $\neg x$ for some $x \in V$. We reduce satisfiability of $F$ to matching with one type variable $\alpha$.
First, we fix a set of type constants $B = V \cup \{\neg x \mid x \in V\}$ and the type constant $\bullet$. Let $\sigma_{x} \equiv \bigcap (B \setminus \{\neg x\})$ and $\sigma_{\neg x} \equiv \bigcap (B \setminus \{x\})$ for $x \in V$.
We construct the set $C$ containing the following constraints
\begin{align*}
&\text{for } x \in V \text{ (consistency)}:\\
&\quad ((\sigma_{\neg x} \to \bullet) \to (\neg x \to \bullet)) \cap ((\sigma_{x} \to \bullet) \to (x \to \bullet)) \dotleq (\alpha \to \bullet) \to (\alpha \to \bullet)\\
&\text{for } (L_1 \vee L_2 \vee L_3) \in F \text{ (validity)}:\\
&\quad(L_1 \to \bullet) \cap (L_2 \to \bullet) \cap (L_3 \to \bullet) \dotleq \alpha \to \bullet
\end{align*}

\noindent If $F$ is satisfied by a valuation $v$, then the substitution $\alpha \mapsto \bigcap\limits_{v(x) = 1} x \cap \bigcap\limits_{v(x) = 0} \neg x$ satisfies $C$.\\
If $C$ is satisfied by a substitution $S$, then by Corollary \ref{cor:path_subtyping} and the consistency constraints we have either $\sigma_{\neg x} \leq S(\alpha) \leq \neg x$ or $\sigma_{x} \leq S(\alpha) \leq x$ for $x \in V$. 
A valuation $v$ constructed according to these cases satisfies each clause in $F$ due to Corollary \ref{cor:path_subtyping} and the validity constraints.

Instead of using constants $\{a_1, \ldots, a_k, \bullet\}$ for an instance of the matching problem, 
encode $[a_i] = \underbrace{\bullet \to \ldots \to \bullet \to}_{i \text{ times}} \bullet$ for $1 \leq i \leq k$ in the proof. Using this technique, 
a single type constant $\bullet$ is sufficient.
\end{proof}

Combining Lemma \ref{lem:matching_one_var} with the reduction in~\cite{RehofEtAlTLCA13} we conclude that neither restricting substitutions to the shape $S : \{\alpha\} \to \TT$ nor restricting to the shape $S : \VV \to \CC$ (atomic substitutions, mapping variables to type constants) reduces the complexity of matching.

\section{Intersection type unification}
\label{sec:unification}

\subsection{The unification problem}
\label{subsec:unification}

\begin{problem}[Satisfiability]
\label{prob:satisfiability}
Given a set of constraints $C = \{\sigma_1 \dotleq \tau_1, \ldots, \sigma_n \dotleq \tau_n\}$,
is there a substitution $S \colon \VV \to \TT$ such that $S(\sigma_i) \leq S(\tau_i)$ for $1 \leq i \leq n$?
\end{problem}

\begin{problem}[Unification]
\label{prob:unification}
Given a set of constraints $C = \{\sigma_1 \doteq \tau_1, \ldots, \sigma_n \doteq \tau_n\}$,
is there a substitution $S \colon \VV \to \TT$ such that $S(\sigma_i) = S(\tau_i)$ for $1 \leq i \leq n$?
\end{problem}

Since for any $\sigma, \tau \in \TT$ and any substitution $S$ we have $S(\sigma) \leq S(\tau) \iff S(\sigma) \cap S(\tau) = S(\sigma)$, satisfiability and unification are (trivially) equivalent. Similarly to matching (resp. one-sided unification) restricting satisfiability (resp. unification) to single constraints does not reduce its complexity.

We now provide a number of observations that give some insight into the type-theoretical and combinatorial expressive power of unification.

Consider a combinatory logic with intersection types \cite{HinSel08,hide92} with arbitrary
basis $\mathfrak{B}$, that is, a finite set of combinator symbols $F,G,\ldots$ with
types $\tau_F, \tau_G, \ldots$. Such a system is given by the rules (applicative fragment)
$\AE, \II, \SUB$ of Definition~\ref{def:bcd} together with a rule assigning types $S(\tau_F)$ to 
the combinator symbol $F$ for any substitution $S$. Write $\mathfrak{B} \vdash E : \tau$ for derivability of the type $\tau$ for the combinatory expression $E$ in this system.

\begin{example}  
Let $\mathfrak{B} = \{F : (\sigma \to \tau) \to a, G : \alpha \to \alpha\}$, where wlog. $\alpha \not\in \Var(\sigma) \cup \Var(\tau)$. In this scenario, type-checking $\mathfrak{B} \vdash F \, G : a$ is equivalent to solving the satisfiability problem $\alpha \to \alpha \dotleq \sigma \to \tau$, or equivalently, the unification problem $\sigma \cap \tau \doteq \sigma$, because we need to find substitutions $S, S_1, \ldots S_n$ for some $n \in \NN$ such that 
\begin{align*}
& \bigcap\limits_{i=1}^n S_i(\alpha \to \alpha)  \leq S(\sigma \to \tau)\\
\stackrel{\text{Lem. } \ref{lem:beta_soundness}}{\iff} & S(\sigma) \leq \bigcap\limits_{i \in I} S_i(\alpha) \text{ and } \bigcap\limits_{i \in I} S_i(\alpha) \leq S(\tau) \text{ for some } I \subseteq \{1, \ldots, n\} \\
\iff & S(\alpha \to \alpha) \leq S(\sigma \to \tau) \text{ setting } S(\alpha) = \bigcap\limits_{i \in I} S_i(\alpha)
\end{align*}
\end{example}%

Write $\mathfrak{B} \vdash^* E : \tau$ if $\mathfrak{B} \vdash E: \tau$ is derivable without the
intersection introduction rule $\II$. This restriction occupies an interesting \enquote{intermediate} position:
generalized to arbitrary bases $\mathfrak{B}$, it is the combinatory logic that subsumes the {\bf BCD}-calculus without intersection introduction
\cite{KT95,RehofU12} and the one-dimensional fragment~\cite{POPL2017} of the intersection type system. For example, using the following basis (cf.~\cite{hide92})
$$\mathfrak{B} = \{ S : (\alpha \to \beta \to \gamma) \to (\alpha' \to \beta) \to \alpha \cap \alpha' \to \gamma, K : \alpha \to \beta \to \alpha, I : \alpha \to \alpha \}$$ the system $\vdash^*$ is sufficient to type $S \, I \, I$, i.e., the SKI-combinatory logic equivalent of the $\lambda$-term $\lambda x. (x \, x)$,
which is not typable in simple types. The full expressive power of $\vdash^*$ is subject to ongoing research as a candidate combinatory logic suitable for software synthesis~\cite{JR13,DMResop14}.

\begin{example}
\label{ex:typability} Typability with respect to $\vdash^*$ is equivalent to unification. Let $\mathfrak{B} = \{F_1 : \tau_1, \ldots F_n : \tau_n\}$, where wlog. $\Var(\tau_i) \cap \Var(\tau_j) = \emptyset$ for $i \neq j$,
and let $E$ be a combinatory term over $\mathfrak{B}$. We want to know whether there is a type $\tau$ such that $\mathfrak{B} \vdash^* E : \tau$. Conveniently, without arrow introduction or intersection introduction, type derivations are entirely term driven up to $\leq$ and initial substitutions.

For any combinatory term $E'$ and type $\tau'$ we define 
$$f(E', \tau') = \begin{cases} 
\{ \tau_i \dotleq \tau' \} & \text{if } E' = F_i \text{ for some } i \in \{1, \ldots, n\}\\
f(E_1, \alpha \to \beta) \cup f(E_2, \alpha) \cup \{\beta \dotleq \tau'\} & \text{if } E' = E_1 \, E_2 \text{ and } \alpha, \beta \text{ are fresh}
\end{cases}$$
Given $E'$ and $\tau'$, by induction on $E'$ we have that the satisfiability problem instance $f(E', \tau')$ has a solution $S$ iff $\mathfrak{B} \vdash^* E' : S(\tau')$. Therefore, $f(E, \alpha)$, where $\alpha$ is fresh, has a solution iff $E$ is typable in the basis $\mathfrak{B}$.

Conversely, given a satisfiability problem $\sigma \dotleq \tau$, we consider typability of
$F \, G$ in the basis $\mathfrak{B} = \{ F: \tau \to a, G: \sigma \}$.
\end{example}
In addition to the previously illustrated non-structurality of unification, the following examples illustrate further potential difficulties.
Example \ref{xmp:non-finitary} shows that unification is not finitary, i.e. unification instances may have infinitely many most general unifiers.
\begin{example}
\label{xmp:non-finitary}
Consider the unification constraint $a \to a \to (\beta \cap b) \doteq  \beta \cap \alpha$ and some solution $S$. We have $S(\beta)=a \to a \to \bigcap\limits_{i \in I} \pi_i$ such that for all $i \in I$ either $\pi_i = b$ or $\pi_i = a \to a \to \pi_j$ for some $j \in I$. Therefore, $S(\beta)$ (and consequently $S(\alpha)$) contains paths that may be arbitrary long, end in the constant $b$ and have an even number of \enquote{$a$}s as arguments. As a result, solutions to the above constraint cannot be built by specialization from a finite set of (most general) unifiers.
\end{example}
The following Example \ref{xmp:exponential-growth} shows that certain unification instances may require all solutions to be of exponential size.
\begin{example}
\label{xmp:exponential-growth}
Consider prime numbers $2, 3$ and the following unification constraints
$$
a \to a \to (\beta_2 \cap b) \doteq  \beta_2 \cap \alpha, \qquad a \to a \to a \to (\beta_3 \cap b) \doteq  \beta_3 \cap \alpha
$$

Let $S$ be any solution to the above constraints. Similar to Example \ref{xmp:non-finitary}, $S$ necessarily satisfies the following properties
\begin{itemize}
\item $S(\beta_2) \neq \omega$, otherwise we would have $S(\alpha) = a \to a \to b$ which by the second constraint implies $a \to a \to a \to (S(\beta_3) \cap b) \leq a \to a \to b$ and therefore $a \to (S(\beta_3) \cap b) \leq b$, which is a contradiction.
\item $S(\beta_3) \neq \omega$, otherwise we would have $S(\alpha) = a \to a \to a \to b$ which by the first constraint implies $a \to a \to (S(\beta_2) \cap b) \leq a \to a \to a \to b$, therefore $S(\beta_2) \leq a \to b$. Again, by the first constraint we obtain $a \to a \to (S(\beta_2) \cap b) \leq a \to b$, which is a contradiction.
\item $S(\alpha) \neq \omega$, otherwise we would have $S(\beta_2) \leq a \to a \to b$, therefore $S(\beta_2) \leq a \to a \to S(\beta_2) \leq a \to a \to a \to a \to b$ and, inductively, $S(\beta_2) \leq a \to \ldots \to a \to b$ with an even number of \enquote{$a$}s, which by Lemma \ref{lem:organized_subtyping} would in total imply a solution of infinite size.
\item $S(\beta_2) = \bigcap_{i \in I} \pi_i$ such that $\pi_i = a \to \ldots \to a \to b$ with an even number of \enquote{$a$}s similar to Example \ref{xmp:non-finitary}. 
\item $S(\beta_3) = \bigcap_{j \in J} \pi_j$ such that $\pi_j = a \to \ldots \to a \to b$ with a number of \enquote{$a$}s that is a multiple of $3$.
\item $S(\alpha) = \bigcap_{k \in K} \pi_k$ such that $\pi_k = a \to \ldots \to a \to b$ with a number of \enquote{$a$}s that is a multiple of $2$ and $3$. 
\end{itemize}

One possible solution $S'$ to the above constraints is
\begin{align*}
S'(\beta_2) &=  (a \to a \to b) \cap (a \to a \to a \to a \to b)\\
S'(\beta_3) &=  a \to a \to a \to b\\
S'(\alpha) &=  a \to a \to a \to a \to a \to a \to b
\end{align*}
In sum, the size of $S(\alpha)$ in any solution is greater than the product of our initial primes $2$ and $3$. By adding an additional constraint $a \to a \to a \to a \to a \to (\beta_5 \cap b) \doteq \beta_5 \cap \alpha$, the size of $S(\alpha)$ is at least $2 \cdot 3 \cdot 5$ in any solution, growing exponentially with additional constraints. As a side note, similarly to Example \ref{xmp:non-finitary}, the above unification instance does not have a most general unifier that produces all possible solutions by specialization.
\end{example}

An axiomatization of the equational theory of intersection type subtyping (without $\omega$) is derived in~\cite{statman2015}. We add two additional axioms  \textbf{(U)} and \textbf{(RE)} in the following Definition \ref{def:aciudlreab} to incorporate the universal type $\omega$.

\begin{definition}[$\textsc{ACIUD}_l\textsc{ReAb}$] 
\label{def:aciudlreab}
The equational theory $\textsc{ACIUD}_l\textsc{ReAb}$ is given by
\begin{description}[before={\renewcommand\makelabel[1]{\upshape\bf ##1}}]
\item[(A)] $\sigma\cap(\tau\cap\rho) \sim (\sigma\cap\tau)\cap\rho$
\item[(C)] $\sigma\cap\tau \sim \tau\cap\sigma$
\item[(I)] $\sigma\cap\sigma \sim \sigma$
\item[(U)] $\sigma \cap\omega \sim \sigma$
\item[(D$_{l}$)] $(\sigma\to\tau)\cap(\sigma\to\tau') \sim \sigma\to\tau\cap\tau'$
\item[(RE)] $\omega \sim \omega \to \omega$
\item[(AB)] $\sigma\to\tau \sim (\sigma\to\tau) \cap (\sigma\cap\sigma'\to\tau)$
\end{description}
\end{definition}

The recursion axiom \textbf{(RE)} captures the recursive nature of $\omega$ and the absorption axiom \textbf{(AB)} captures contra-variance. 
\begin{lemma}
\label{lem:sim-is-eq}
Given $\sigma, \tau \in \TT$ we have $\sigma = \tau$ iff $\sigma \sim \tau$.
\end{lemma}

\begin{proof}
\enquote{$\pmb{\Longrightarrow}$}: Induction on the depth of the derivation of $\sigma \leq \tau$ to show $\sigma \cap \tau \sim \sigma$. The most interesting subcase is $\sigma_1 \to \tau_1 \leq \sigma_2 \to \tau_2$ under the assumptions $\sigma_2 \leq \sigma_1$ and $\tau_1 \leq \tau_2$. In the subcase we have
\begin{equation*}
  \begin{array}{r@{\,}c@{\,}l}
\sigma_1 \to \tau_1 & \stackrel{(\text{AB})}{\sim} & (\sigma_1 \to \tau_1) \cap (\sigma_1 \cap \sigma_2 \to \tau_1) \\
& \stackrel{\text{IH}}{\sim} & (\sigma_1 \to \tau_1) \cap (\sigma_1 \cap \sigma_2 \to \tau_1 \cap \tau_2) \\
& \stackrel{(\text{D}_l)}{\sim}&  (\sigma_1 \to \tau_1) \cap (\sigma_1 \cap \sigma_2 \to \tau_1) \cap (\sigma_1 \cap \sigma_2 \to \tau_2) \\
& \stackrel{\text{IH}}{\sim}&  (\sigma_1 \to \tau_1) \cap (\sigma_1 \cap \sigma_2 \to \tau_1) \cap (\sigma_2 \to \tau_2) \\
& \stackrel{(\text{AB})}{\sim}&  (\sigma_1 \to \tau_1) \cap (\sigma_2 \to \tau_2)
\end{array}
\end{equation*}

Finally, $\sigma \leq \tau$ and $\tau \leq \sigma$ imply $\sigma \sim \sigma \cap \tau \sim \tau$.\\
\enquote{$\pmb{\Longleftarrow}$}: Each axiom of $\textsc{ACIUD}_l\textsc{ReAb}$ is derivable using subtyping. Again, the most interesting subcase is $\sigma\to\tau \sim (\sigma\to\tau) \cap (\sigma\cap\sigma'\to\tau)$. First, by contravariance we have $\sigma\to\tau \leq \sigma\cap\sigma'\to\tau$, therefore $\sigma\to\tau \leq (\sigma\to\tau) \cap (\sigma\cap\sigma'\to\tau)$. Second, $(\sigma\to\tau) \cap (\sigma\cap\sigma'\to\tau) \leq \sigma\to\tau$. Therefore, $\sigma\to\tau = (\sigma\to\tau) \cap (\sigma\cap\sigma'\to\tau)$.
\end{proof}

The absorption axiom \textbf{(AB)} distinguishes the above theory $\textsc{ACIUD}_l\textsc{ReAb}$ from theories studied in literature. 
As discussed in the introduction, the closest equational theory $\textsc{ACIUD}_l$ of \cite{anantharaman2004}, which assumes $\omega \to \sigma \sim \omega \sim \sigma \to \omega$ and has no 
equivalent of the absorption axiom \textbf{(AB)}, is \textsc{Exptime}-complete. Unfortunately, the absorption axiom prevents the approaches presented in~\cite{anantharaman2004, anantharaman2003acid} as shown by the following examples.

\begin{example}
\label{xmp:circular_dependency}
Consider $\alpha \cap (\alpha \to a) \doteq \alpha$ (or equivalently $\alpha \dotleq \alpha \to a$). A DAG-based (or \enquote{occurs-check}-based) approach cannot stratify such a constraint since any solution $S(\alpha)$ contains at least one subterm $S(\alpha) \to a$ and therefore a circular dependency.
Interestingly, using absorption there is a solution $S(\alpha) = \omega \to a$.
\end{example}

\begin{example}
Consider $\alpha \cap (((\alpha \to c) \cap b) \to a) \doteq \alpha$ (or equivalently $\alpha \dotleq ((\alpha \to c) \cap b) \to a$). In contrast to the previous example, all occurrences of $\alpha$ are positive. Again, we have a circular dependency. Using absorption there is a solution $S(\alpha) = b \to a$.
\end{example}

The constraint $\alpha \cap (\alpha \to a) \doteq \alpha$ (or equivalently $\alpha \dotleq \alpha \to a$) in Example \ref{xmp:circular_dependency} has several interesting properties. For any $\sigma$ such that $\sigma \leq \sigma \to a$ we may type the (Church-style annotated) lambda-term $\lambda x^{\sigma}.(x^{\sigma \to a}\;x^{\sigma})$ using subtyping. Additionally, there are several structurally different solution for the above constraint:
\begin{itemize}
\item $S_1(\alpha) = \omega \to a$ (requires \textbf{(AB)}, \textbf{(U)}, \textbf{(C)})
\item $S_2(\alpha) = a \cap (a \to a)$ (requires \textbf{(AB)})
\item $S_3(\alpha) = ((a \cap (a \to a)) \to a) \to a$ (requires \textbf{(AB)}, \textbf{(D$_l$)}, \textbf{(A)}, \textbf{(C)})
\end{itemize}
$S_1$ requires the absorption axiom as well as the unit axiom, whereas $S_2$ and $S_3$ do not require any axiom containing $\omega$. Therefore, $\omega$ may play a crucial role in solving circular constraints but is not strictly necessary. The main difference between $S_2$ and $S_3$ is that $S_3$ uses distributivity and, more importantly, only one path (or factor in the sense of~\cite{statman2015}) instead of an intersection of two paths. This is especially interesting in the light of Lemma~\ref{lem:organized_subtyping}. In particular, restricting substitutions to paths does not avoid satisfiable circular constraints.

\subsection{Tiling games}
\label{sec:tiling-games}
In this section we introduce a special kind of domino tiling game, referred to as two-player corridor tiling games, for which the problem of existence of winning strategies is
\textsc{Exptime}-complete \cite{chlebus1986domino}. We then show that \textsc{Exptime}-completeness is preserved when tilings are restricted to a particular \enquote{spiral} shape, which will be used to prove
our \textsc{Exptime}-lower bound for intersection type unification in Sec.~\ref{sec:exptime-lb}.

\begin{definition}[Tiling System]
A \emph{tiling system} is a tuple $(D,H,V,\bar{b},\bar{t},n)$, where
\begin{itemize}
\item $D$ is a finite set of tiles (also called dominoes)
\item $H,V \subseteq D \times D$ are horizontal and vertical constraints
\item $\bar{b},\bar{t} \in D^n$ are $n$-tuples of tiles
\item $n$ is a unary encoded natural number
\end{itemize}
\end{definition}

\begin{definition}[Corridor Tiling]
Given a tiling system $(D,H,V,\bar{b},\bar{t},n)$, a \emph{corridor tiling} is a mapping $\lambda : \{1, \ldots, l\} \times \{1, \ldots, n\} \to D$ for some $l \in \NN$ such that
\begin{itemize}
\item $\bar{b} = (\lambda(1,1), \ldots, \lambda(1, n))$ (correct bottom row)
\item $\bar{t} = (\lambda(l,1), \ldots, \lambda(l, n))$ (correct top row)
\item for $i \in \{1, \ldots, l\}$ and $j \in \{1, \ldots, n-1\}$ we have $(\lambda(i,j), \lambda(i,j+1)) \in H$, i.e. the horizontal constraints are satisfied
\item for $i \in \{1, \ldots, l-1\}$ and $j \in \{1, \ldots, n\}$ we have $(\lambda(i,j), \lambda(i+1,j)) \in V$, i.e. the vertical constraints are satisfied
\end{itemize}
\end{definition}

Given a tiling system $(D,H,V,\bar{b},\bar{t},n)$, a \emph{Two-Player Corridor Tiling} game consists of two players (\emph{Constructor} and \emph{Spoiler}). 
The game is played on an $\NN \times \{1, \ldots, n\}$ board and starts with the bottom row $\bar{b}$.
Each player places copies of tiles in turn starting with Constructor. 
While Constructor tries to construct a corridor tiling, Spoiler tries to prevent it. Constructor wins if Spoiler makes an illegal move (with respect to $H$ or $V$), or when a correct corridor tiling is completed. 
We say Constructor has winning strategy, if he can win no matter what Spoiler does.

\begin{lemma}[Chlebus \cite{chlebus1986domino}]
The decision problem whether Constructor has a winning strategy in a given two-player corridor tiling game is \textsc{Exptime}-complete.
\end{lemma}
Instead of directly encoding a Two-Player Corridor Tiling into intersection type satisfiability, we introduce a slightly different game that is played out as sequences instead of corridors.
The main goal is to get rid of several structural constraints of corridors for a more accessible construction of a spiral where each new tile has a neighboring previous tile.

\begin{definition}[Spiral Tiling]
Given a tiling system $(D,H,V,\bar{b},\bar{t},n)$, a \emph{spiral tiling} is a sequence $d_1 \ldots d_m \in D^m$ for some $m \in \NN$ such that 
\begin{itemize}
\item $\bar{b} = d_1 \ldots d_n$ (correct prefix)
\item $\bar{t} = d_{m-n+1} \ldots d_m$ (correct suffix)
\item $(d_i, d_{i+1}) \in H$ for $1 \leq i \leq m-1$ (horizontal constraints)
\item $(d_i, d_{i+n}) \in V$ for $1 \leq i \leq m-n$ (vertical constraints)
\end{itemize}
\end{definition}

Given a tiling system $(D,H,V,\bar{b},\bar{t},n)$ a \emph{Two-Player Spiral Tiling} game, played by Constructor and Spoiler, starts with the sequence $\bar{b}$. Each player adds a copy of a tile to the end of the current sequence taking turns starting with Constructor. While Constructor tries to construct a spiral tiling, Spoiler tries to prevent it. 
Constructor wins if Spoiler makes an illegal move (with respect to $H$ or $V$), or when a correct spiral tiling is completed. Again, we are interested in whether Constructor has a winning strategy.

The main difference between a corridor tiling and a spiral tiling is the lack of individual rows. While a tile at the beginning of the new row of a corridor is not constrained by the previously placed tile, in a spiral each new tile is constrained by the previous one. Additionally, any corridor tiling contains $l \cdot n$ tiles for some $l$; a spiral tiling does not obey this restriction. To clarify essential aspects of spiral tilings, consider the following examples.

\begin{example}
Consider the tiling system with $D = \{a,b\}$, $H = \{(a,b),(b,a),(b,b)\}$, $V = D^2$, $\bar{b} = aaa$, $\bar{t} = bbb$ and $n= 3$. Constructor does not have a winning strategy in the corresponding Two-Player Spiral Tiling game. During the game on Spoiler's turn there are two possibilities. In case the current sequence ends in $a$, Spoiler is forced/allowed to append $b$, which does not result in the suffix $bbb$. Regardless of Constructor's next tile the suffix is not $bbb$. In case the current sequence ends in $b$, Spoiler is allowed to append $a$ and, similarly, Constructor is not able to produce a spiral tiling.
\end{example}


\begin{example}
Consider the tiling system with $D = \{a,b\}$, $H = D^2$, $V = D^2 \setminus \{(b,a)\}$, $\bar{b} = aaaaa$, $\bar{t} = bbbbb$ and $n=5$. Constructor has a winning strategy in the corresponding Two-Player Spiral Tiling game by always appending $b$. Due to $V$, if the current sequence has the tile $b$ at position $i$, then it will have the tile $b$ at any later position $i+5j$ for $j \in \NN$. Therefore, after the first nine turns of the game all positions $1+5j$, $3+5j$, $5+5j$, $7+5j$ and $9+5j$ for $j \in \NN$ will have the tile $b$ regardless of Spoiler's moves. Since those positions will form a suffix $bbbbb$ after $9$ turns, Constructor is able to produce a spiral tiling.
\end{example}

\begin{lemma}
\label{lem:spiral}
The decision problem whether Constructor has a winning strategy in a given two-player spiral tiling game is \textsc{Exptime}-complete.
\end{lemma}

\begin{proof}\ \\
{\bf Lower Bound:}
Given a tiling system $T = (D,H,V,(b_1, \ldots, b_n),(t_1, \ldots, t_n) ,n)$, let 
\begin{align*}
D' &= D \mathbin{\dot{\cup}} \{\#\}\\
H' &= H \mathbin{\dot{\cup}} \{(d, \#) \mid d \in D'\} \cup \{(\#, d) \mid d \in D'\}\\
V' &= V \mathbin{\dot{\cup}} \{(\#, \#)\}\\
T' &= (D', H', V', (b_1, \ldots, b_n, \#, \#),(t_1, \ldots, t_n, \#, \#),n+2)
\end{align*}
We show that Constructor has a winning strategy for Two-Player Corridor Tiling in $T$ iff he has a winning strategy for Two-Player Spiral Tiling in $T'$.

By construction, both players are allowed to and have to place the tile $\#$ at exactly the turns $i(n+2)-1$ and $i(n+2)$ for $i \geq 1$. Therefore, a winning strategy does not branch nor end at those turns. Additionally, a correct spiral tiling ends in two consecutive $\#$ tiles, therefore necessarily contains $i(n+2)$ tiles.

From any correct corridor tiling $\lambda : \{1, \ldots, l\} \times \{1, \ldots, n\}$ for $T$ we construct a spiral tiling $d_1 \ldots d_{l(n+2)}$ for $T'$ by 
$$d_k = \begin{cases}
\lambda(i,j) & \text{if } k = (i-1)(n+2)+j \text{ and } i \geq 1 \text{ and } 1 \leq j \leq n \\
\# & \text{if } k = (i-1)(n+2)+j \text{ and } i \geq 1 \text{ and either } j = 0 \text{ or } j = n+1
\end{cases}$$

From a correct spiral tiling $d_1 \ldots d_{l(n+2)}$ for $T'$ we construct a corridor tiling $\lambda : \{1, \ldots, l\} \times \{1, \ldots, n\}$ for $T$ by
$\lambda(i,j) = d_{(i-1)(n+2)+j}$.

In particular, Constructor's winning strategy (skipping/adding the forced $\#$ turns) is exactly the same for both games.

{\bf Upper Bound:} Computation in \textsc{Apspace} = \textsc{Exptime} (similar to Two-Player Corridor Tiling). To continue, only the $n$ previously placed tiles have to be considered.
\end{proof}

\subsection{\textsc{Exptime} lower bound}
\label{sec:exptime-lb}

We now prove our main result, that the intersection type unification problem is
\textsc{Exptime}-hard. The proof will be by reduction from spiral tiling games
(Lemma~\ref{lem:spiral}) to the intersection type satisfiability problem.

Let $T = (D,H,V,\bar{b}= b_1 \ldots b_n,\bar{t} = t_1 \ldots t_n,n)$ be a tiling system. Wlog. $(b_i, b_{i+1}) \in H$ and $(t_i, t_{i+1}) \in H$ for $1 \leq 1 < n$, otherwise the given prefix $\bar{b}$ or the given suffix $\bar{t}$ would already violate constraints on consecutive tiles and Constructor would have no chance to construct a spiral tiling. We fix the set of type constants $\CC = D \mathop{\dot{\cup}} \{\bullet\}$ and variables $\VV = \{\alpha\} \cup \{\beta_d \mid d \in D\}$ and construct the following set of constraints $\mathscr{C}_T$:

\begin{enumerate}[label=(\roman*)]
\item $\sigma_\bot^H \cap \sigma_\bot^V \cap \sigma_t \cap \bigcap\limits_{d \in D} \beta_d \dotleq \sigma_b \cap \bigcap\limits_{d' \in D} \bigcap\limits_{d \in D} (d' \to d \to \beta_d)$ \hfill (Game moves)
\item $\bigcap\limits_{(d', d) \in H} (d \to d' \to \alpha) \dotleq \bigcap\limits_{d \in D} (d \to \beta_d)$ \hfill ($d$ respects $H$)
\item $\bigcap\limits_{(d', d) \in V} (d \to \underbrace{\omega \to \ldots \to \omega \to}_{n-1 \text{ times}} d' \to \alpha) \dotleq \bigcap\limits_{d \in D} (d \to \beta_d)$ \hfill ($d$ respects $V$)
\end{enumerate}
where
\begin{align*}
\sigma_b & \equiv b_n \to \ldots \to b_1 \to \bullet & \text{(Initial state)} \\
\sigma_t & \equiv (t_n \to \ldots \to t_1 \to \alpha) \cap (\omega \to t_n \to \ldots \to t_1 \to \alpha) & \text{(Final states)} \\
\sigma_\bot^H & \equiv \!\!\!\bigcap\limits_{(d,d') \in D \times D \setminus H} \!\!\!(d' \to d \to \alpha) & \text{($d'$ violates $H$)} \\
\sigma_\bot^V & \equiv \!\!\!\bigcap\limits_{(d,d') \in D \times D \setminus V} \!\!\!(d' \to  \underbrace{\omega \to \ldots \to \omega \to}_{n-1 \text{ times}} d \to \alpha) & \text{($d'$ violates $V$)}
\end{align*}

Intuitively, we use Lemma \ref{lem:organized_subtyping} to realize alternation. The rhs of $(i)$ represents an intersection of all board positions which Constructor may face. For all such positions he needs to find a suitable move by choosing a path on the lhs of $(i)$. He can either state that the Spoiler's last move violates $H$ (resp. $V$) choosing $\sigma_\bot^H$ (resp. $\sigma_\bot^V$), or  that the game is finished choosing $\sigma_t$, 
or he can pick his next move $d \in D$ choosing $\beta_d$. Intuitively, $\beta_d$ captures all board positions in which Constructor decides to place $d$ next.
Note that on the rhs of $(i)$ in the type $d' \to d \to \beta_d$ the tile $d'$ is not constrained (representing all possible Spoiler's moves) while the tile $d$ is constrained to the index of $\beta_d$, i.e. Constructor's previous choice. Therefore, by picking a move $d$ Constructor faces board positions that arise from the previous position extended by $d$ and each possible $d'$.
Constraints $(ii)$ and $(iii)$ ensure that whenever Constructor picks his next move $d \in D$ choosing $\beta_d$ he has to respect $H$ and $V$.

We show that Constructor has a winning strategy for two-player spiral tiling in $T$ iff the constraint system $\mathscr{C}_T$ is satisfiable.
To represent game positions as types, we define the mapping $[\cdot] : D^* \to \TT$ such that $[\epsilon] = \bullet$ and $[\bar{s}d] = d \to [\bar{s}]$ for $\bar{s} \in D^*$ and $d \in D$. To improve readability, we use the notation $\sigma \stackrel{\phi}{\leq} \tau$, where $\phi$ is a hint why the inequality holds.

\begin{lemma} Let $T$ be a tiling system. If Constructor has a winning strategy in a two-player spiral tiling game in $T$, then the constraint system $\mathscr{C}_T$ is satisfiable.
\label{lem:lr}
\end{lemma}
\begin{proof}
Assume that Constructor has a winning strategy that is represented by a labeled tree $f : \dom(f) \to \{C,S\}$ where 
\begin{itemize}
\item $\dom(f) \subseteq D^*$ is finite and prefix-closed, i.e. $\bar{u}\bar{v} \in \dom(f)$ implies $\bar{u} \in \dom(f)$.
\item $\depth(f) = \max\{k \mid d_1 \ldots d_k \in \dom(f)\}$.
\item For $\bar{s} = d_1 \ldots d_k \in \dom(f)$ we have $f(\bar{s}) = C$ if $k$ is even and $f(\bar{s}) = S$ if $k$ is odd, i.e. $C$  (resp. $S$) places a tile after an even (resp. odd) number of turns.
\item For $\bar{s} \in \dom(f)$ such that $f(\bar{s}) = S$ we have $\bar{s}d' \in \dom(f)$ for all $d' \in D$, i.e. the strategy has to consider all (possibly illegal) Spoiler's moves.
\item For $\bar{s} \in \dom(f)$ such that $f(\bar{s}) = C$ we have either 
\begin{itemize}
\item There exists exactly one $d \in D$ such that $\bar{s}d \in \dom(f)$ and $\bar{b} \bar{s} = \bar{u} d_1 \ldots d_n$ for some $\bar{u} \in D^*$ and $d_1, \ldots, d_n \in D$ with $(d_n,d) \in H$ and $(d_1, d) \in V$, i.e. Constructor's next move is $d$ which respects $H$ and $V$.
\item $\bar{s}d \not\in \dom(f)$ for all $d \in D$ and either
\begin{itemize}
\item $\bar{b} \bar{s} = \bar{u} \bar{t}$ for some $\bar{u} \in D^*$, i.e. Constructor states that the game is finished.
\item $\bar{b} \bar{s} = \bar{u} \bar{t} d'$ for some $\bar{u} \in D^*$ and $d' \in D$, i.e. Constructor states that Spoiler's last move $d'$ is illegal because the game already ended.
\item $\bar{b} \bar{s} = \bar{u} d d'$ for some $\bar{u} \in D^*$, $d,d' \in D$ such that $(d,d') \not\in H$, i.e. Constructor states that Spoiler's last move $d'$ violates $H$.
\item $\bar{b} \bar{s} = \bar{u} d \bar{v} d'$ for some $\bar{u} \in D^*$, $\bar{v} \in D^{n-1}$, $d,d' \in D$ such that $(d,d') \not\in V$, i.e. Constructor states that Spoiler's last move $d'$ violates $V$.
\end{itemize}
\end{itemize}
\end{itemize}
We construct the following substitution $S$
$$
S(\alpha)  = \bigcap_{d_1 \ldots d_k = \bar{s} \in D^* \atop k \leq \depth(f)+n} [\bar{s}] \quad \text{ and } \quad
S(\beta_d)  = \bigcap_{\bar{s} \in f^{-1}(C) \atop \bar{s}d \in \dom(f)} [\bar{b}\bar{s}] \quad \text{ for } d \in D\\
$$
We verify that the individual inequalities hold.
\begin{itemize}
\item $S(\sigma_\bot^H \cap \sigma_\bot^V \cap \sigma_t \cap \bigcap\limits_{d \in D} \beta_d) \leq \sigma_b$: \\
if $\bar{b} = \bar{t}$, then $S(\sigma_t) \leq \sigma_b$.
Otherwise, according to $f$, there exists a $d \in D$ such that $d \in \dom(f)$. Therefore, $S(\beta_d) \leq [\bar{b}] \equiv \sigma_b$.
\item $S(\sigma_\bot^H \cap \sigma_\bot^V \cap \sigma_t \cap \bigcap\limits_{d \in D} \beta_d) \leq S(d' \to d \to \beta_d)$ for all $d,d' \in D$: \\
we show that for any $\pi = d_k \to \ldots \to d_1 \to \sigma_b = [\bar{b}\bar{s}]$ such that $d_1 \ldots d_k = \bar{s} \in f^{-1}(C)$ and $\bar{s}d \in \dom(f)$ we have $S(\sigma_\bot^H \cap \sigma_\bot^V \cap \sigma_t \cap \bigcap\limits_{d \in D} \beta_d) \leq d' \to d \to \pi \equiv [\bar{b}\bar{s}dd']$.
Since $\bar{s}d \in \dom(f)$ and $f(\bar{s}) = C$ we have $\bar{s}dd' \in \dom(f)$ and $f(\bar{s}dd') = C$. According to $f$ we have either
\begin{itemize}
\item $\bar{s}dd'd'' \in \dom(f)$ for some $d'' \in D$. Therefore, $S(\beta_{d''}) \leq d' \to d \to \pi \equiv [\bar{b}\bar{s}dd']$
\item or
\begin{itemize}
\item $\bar{b}\bar{s}dd' = \bar{u} \bar{t}$ for some $\bar{u} \in D^*$, then $S(\sigma_t) \leq [\bar{u}\bar{t}] \equiv [\bar{b}\bar{s}dd']$.
\item $\bar{b}\bar{s}dd' = \bar{u} \bar{t} d'$ for some $\bar{u} \in D^*$, then $S(\sigma_t) \leq [\bar{u} \bar{t} d'] \equiv [\bar{b}\bar{s}dd']$.
\item $(d,d') \not\in H$, then $S(\sigma_\bot^H) \leq [\bar{b}\bar{s}dd']$.
\item $\bar{b} \bar{s} d d' = \bar{u} d'' \bar{v} d'$ for some $\bar{u} \in D^*$, $\bar{v} \in D^{n-1}$, $d'' \in D$ such that $(d'',d') \not\in V$, then $S(\sigma_\bot^V) \leq [\bar{u} d'' \bar{v} d'] \equiv [\bar{b} \bar{s} d d']$.
\end{itemize}
\end{itemize}
\item $S(\bigcap\limits_{(d', d) \in H} (d \to d' \to \alpha)) \leq S(d \to \beta_d)$ for all $d \in D$: \\
fix any $\pi = d_k \to \ldots \to d_1 \to \sigma_b = [\bar{b}\bar{s}]$ such that $d_1 \ldots d_k = \bar{s} \in f^{-1}(C)$ and $\bar{s}d \in \dom(f)$. Since $f(\bar{s}) = C$ and $\bar{s}d \in \dom(f)$, we have $\bar{b}\bar{s}d = \bar{u}d'd$ for some $\bar{u} \in D^*$ and $d' \in D$ such that $(d',d) \in H$. We have 
$$S(\bigcap\limits_{(d', d) \in H} (d \to d' \to \alpha)) \stackrel{S(\alpha) \leq [\bar{u}]}{\leq} [\bar{u}d'd] \equiv [\bar{b}\bar{s}d] \equiv d \to \pi$$
\item $S(\bigcap\limits_{(d', d) \in V} (d \to \underbrace{\omega \to \ldots \to \omega \to}_{n-1 \text{ times}} d' \to \alpha)) \leq S(d \to \beta_d)$ for all $d \in D$:\\
fix any $\pi = d_k \to \ldots \to d_1 \to \sigma_b = [\bar{b}\bar{s}]$ such that
$d_1 \ldots d_k = \bar{s} \in f^{-1}(C)$ and $\bar{s}d \in \dom(f)$. Since
$f(\bar{s}) = C$ and $\bar{s}d \in \dom(f)$, we have $\bar{b}\bar{s}d =
\bar{u}d'\bar{v}d$ for some $\bar{u} \in D^*$, $\bar{v} \in D^{n-1}$ and $d' \in
D$ such that $(d',d) \in V$. We have
\[
S(\bigcap\limits_{(d', d) \in V} (d \to \underbrace{\omega \to \ldots \to \omega
  \to}_{n-1 \text{ times}} d' \to \alpha)) \stackrel{S(\alpha) \leq
  [\bar{u}]}{\leq} [\bar{u}d'\bar{v}d] \equiv [\bar{b}\bar{s}d] \equiv d \to \pi
\tag*{\qEd}
\]
\end{itemize}
\def\popQED{}
\end{proof}

\begin{lemma} Let $T$ be a tiling system. If the constraint system $\mathscr{C}_T$ is satisfiable, then Constructor has a winning strategy in a two-player spiral tiling game in $T$. 
\label{lem:rl}
\end{lemma}
\begin{proof}
Assume that there exists a substitution $S$ that satisfies the constraints 
$\mathscr{C}_T$ and wlog. uses only organized types. Constructor wins the game regardless of Spoiler's moves as follows: 
Let $\tau \equiv S(\sigma_b \cap \bigcap\limits_{d' \in D} \bigcap\limits_{d \in D} (d' \to d \to \beta_d))$, i.e. the rhs of $(i)$.
The initial game position is $\bar{b}$ and we have $\tau \leq \sigma_b \equiv [\bar{b}]$. We consider a single turn of Constructor from any game position $\bar{s} \in D^*$ that he may face.\\
Assume $(\star)$ that the current game position $\bar{s}$ satisfies $\tau \leq [\bar{s}]$.
Due to $(i)$ and Corollary \ref{cor:path_subtyping} we have the following cases
\begin{itemize}
\item If $S(\sigma_\bot^H) \leq [\bar{s}]$, then there exist $d,d' \in D$ such that $(d,d') \not\in H$ and for some path $\pi$ we have $d' \to d \to \pi \leq [\bar{s}]$. Therefore, $\bar{s} = \bar{u}dd'$ for some $\bar{u} \in D^*$ and Constructor wins because Spoiler's last move $d'$ violates $H$. As a side note, his is not possible for $\bar{s} = \bar{b}$ since $\bar{b}$ is horizontally consistent.
\item If $S(\sigma_\bot^V) \leq [\bar{s}]$, then there exist $d,d' \in D$ such that $(d,d') \not\in V$ and for some path $\pi$ we have $d' \to  \underbrace{\omega \to \ldots \to \omega \to}_{n-1 \text{ times}} d \to \pi \leq [\bar{s}]$. Therefore, $\bar{s} = \bar{u}d\bar{v}d'$ for some $\bar{u} \in D^*$, $\bar{v} \in D^{n-1}$ and Constructor wins because Spoiler's last move $d'$ violates $V$. As a side note, this is not possible for $\bar{s} = \bar{b}$ since $\bar{b}$ is too short.
\item If $S(\sigma_t) \leq [\bar{s}]$, then Constructor wins because $t_n \to \ldots \to t_1 \to \pi \leq [\bar{s}]$ for some path $\pi$ implies the winning condition $\bar{s} = \bar{u}\bar{t}$ for some $\bar{u} \in D^*$. Alternatively $\omega \to t_n \to \ldots \to t_1 \to \pi \leq [\bar{s}]$ for some path $\pi$ implies $\bar{s} = \bar{u}\bar{t}d'$ for some $\bar{u} \in D^*$ and $d' \in D$, therefore Spoiler's last move $d'$ was illegal because the game already ended.
\item If $S(\beta_d) \leq [\bar{s}]$ for some $d \in D$, then Constructor may safely place $d$ as the next tile. We verify consistency wrt. $H$ and $V$. First, due to $(ii)$ there exists a path $d \to d' \to \pi$ for some $d' \in D$ with $(d', d) \in H$ such that 
$$d \to d' \to \pi \stackrel{(ii)}{\leq} S(d \to \beta_d) \stackrel{S(\beta_d) \leq [\bar{s}]}{\leq} [\bar{s}d]$$
Therefore, $\bar{s}d = \bar{u}d'd$ for some $\bar{u} \in D^*$ and placing $d$ does not violate $H$.
Second, due to $(iii)$ there exists a path $d \to \underbrace{\omega \to \ldots \to \omega \to}_{n-1 \text{ times}} d' \to \pi'$ for some $d' \in D$ with $(d', d) \in V$ such that 
$$d \to \underbrace{\omega \to \ldots \to \omega \to}_{n-1 \text{ times}} d' \to \pi' \stackrel{(iii)}{\leq} S(d \to \beta_d) \stackrel{S(\beta_d) \leq [\bar{s}]}{\leq} [\bar{s}d]$$
Therefore, $\bar{s}d = \bar{u}d'\bar{v}d$ for some $\bar{u} \in D^*$, $\bar{v} \in D^{n-1}$ and placing $d$ does not violate $V$.
\end{itemize}
In neither case Constructor loses. If Constructor placed the tile $d \in D$, in which case we have $S(\beta_d) \leq [\bar{s}]$, Spoiler may place any tile $d' \in D$. The new game position is $\bar{s}dd'$ and our initial assumption $(\star)$ is inductively satisfied $$\tau \leq S(d' \to d \to \beta_d) \stackrel{S(\beta_d) \leq [\bar{s}]}{\leq} [\bar{s}dd']$$
Therefore, we may apply our argument inductively.
Additionally, the game necessarily ends after a finite number of turns: if $\tau = \bigcap\limits_{i \in I} (\sigma_1^i \to \ldots \sigma_{l_i}^i \to c_i)$ (for some index set $I$, integers $l_i \geq 0$ for $i \in I$, types $\sigma_j^i$ for $i \in I$ and $1 \leq j \leq l_i$ and type constants $c_i$ for $i \in I$), then $(\star)$, i.e. $\tau \leq [\bar{s}]$, cannot be satisfied by any $\bar{s} \in D^k$ with $k > \max\{l_i \mid i \in I\}$.
\end{proof}

\begin{theorem}
\label{thm:lb}
The intersection type satisfiability problem and the intersection type unification problem
are \textsc{Exptime}-hard.
\end{theorem}
\begin{proof}
Immediate from Lemma~\ref{lem:spiral}, Lemma~\ref{lem:lr} and Lemma~\ref{lem:rl}, since all reduction steps are evidently computable in polynomial time. Moreover, satisfiability is polynomial time equivalent to unification.
\end{proof}

\begin{corollary}
Satisfiability and unification are \textsc{Exptime}-hard even in the presence of only one constant.
\end{corollary}


\begin{proof}
Encode $d_i$ by $[d_i] = \underbrace{\bullet \to \ldots \to \bullet \to}_{i \text{ times}} \bullet$ for $1 \leq i \leq k$ in the original proofs.
\end{proof}


\begin{corollary}
\label{cor:codomain-restriction}
Satisfiability and unification are \textsc{Exptime}-hard even if the codomain of substitutions is restricted to types of rank 1, i.e. intersections of simple types.
\end{corollary}

\begin{proof}
In the proof of Lemma \ref{lem:lr} each variable is substituted by an intersection of simple types, i.e. a rank 1 intersection type.
\end{proof}

Interestingly, the axiom \textbf{(RE)} $\omega \sim \omega \to \omega$ (resp. $\omega \leq \omega \to \omega$) is not necessary for the \textsc{Exptime} lower bound proof, while the axioms \textbf{(U)} and \textbf{(AB)} (resp. $\sigma \leq \omega$ and $\omega \to \tau \leq \sigma \to \tau$ derived from contravariance) play a crucial role in the construction of $\sigma_\bot^V$ to capture an exponential number of cases. Naturally, we would like to know whether the special properties of $\omega$ are necessary for the \textsc{Exptime} lower bound, which we will consider in the following Sec. \ref{sec:exptime-lb-no-omega}.

\subsection{\textsc{Exptime} lower bound without $\omega$}
\label{sec:exptime-lb-no-omega}
In AC unification theories~\cite{Baader01} it is often the case that presence of unit does not change unification complexity. This holds true for the \textsc{Exptime} lower bound of satisfiability in the $\omega$-free fragment of intersection types and the corresponding $\textsc{ACID}_l\textsc{Ab}$ unification. We observe that in the construction of constraints $\mathscr{C}_T$ in Section \ref{sec:exptime-lb} it suffices to consider substitutions that map variables to intersections of types of the shape $d_1 \to \ldots \to d_k \to \bullet$ where $d_1, \ldots, d_k \in D$. Therefore, we may replace the use of $\omega$ in  $\mathscr{C}_T$ by iterated distributivity. Generally, this is possible because the codomain of substitutions necessary for the lower bound construction may be restricted to a fixed rank.

Given a tiling system $T = (D,H,V,\bar{b}= b_1 \ldots b_n,\bar{t} = t_1 \ldots t_n,n)$ we fix the set of type constants $\CC = D \mathop{\dot{\cup}} \{\bullet\}$ and variables $\VV = \{\alpha\} \cup \{\beta_d \mid d \in D\} \cup \{\gamma_d^i \mid d \in D, i \in \{1, \ldots, n\}\}$ and construct the following set of constraints $\mathscr{C}_T'$:

\begin{enumerate}[label=(\roman*)]
\item $\sigma_\bot^H \cap \sigma_\bot^V \cap \sigma_t \cap \bigcap\limits_{d \in D} \beta_d \dotleq \sigma_b \cap \bigcap\limits_{d' \in D} \bigcap\limits_{d \in D} (d' \to d \to \beta_d)$ \hfill (Game moves)
\item $\sigma_\top^H \dotleq \bigcap\limits_{d \in D} (d \to \beta_d)$ \hfill ($d$ respects $H$)
\item $\sigma_\top^V \dotleq \bigcap\limits_{d \in D} (d \to \beta_d)$ \hfill ($d$ respects $V$)
\item $\gamma_d^1 \doteq d \to \alpha$ \hfill (for each $d \in D$, $d$ is last)
\item $\gamma_d^i \doteq \bigcap\limits_{d_i \in V} (d_i \to \gamma_d^{i-1})$ \hfill (for each $d \in D$ and each $i=2 \ldots n$, $d$ is $i$-th to last)
\end{enumerate}
\begin{align*}
\text{where } \qquad \sigma_b & \equiv b_n \to \ldots \to b_1 \to \bullet & \text{(Initial state)}\\
\sigma_t & \equiv (t_n \to \ldots \to t_1 \to \alpha) \cap \bigcap\limits_{d' \in D}(d' \to t_n \to \ldots \to t_1 \to \alpha) & \text{(Final states)}\\
\sigma_\bot^H & \equiv \bigcap\limits_{(d,d') \in D \times D \setminus H} (d' \to d \to \alpha) & \text{($d'$ violates $H$)}\\
\sigma_\bot^V & \equiv \bigcap\limits_{(d,d') \in D \times D \setminus V} (d' \to \gamma_{d}^{n}) & \text{($d'$ violates $V$)}\\
\sigma_\top^H & \equiv \bigcap\limits_{(d', d) \in H} (d \to d' \to \alpha) & \text{($d$ respects $H$)}\\
\sigma_\top^V & \equiv \bigcap\limits_{(d', d) \in V} (d \to \gamma_{d'}^{n}) & \text{($d$ respects $V$)}
\end{align*}

The above constraints $\mathscr{C}_T'$ differ from $\mathscr{C}_T$ in two aspects. First, final states in $\sigma_t$ are described explicitly for each possible last tile $d' \in D$, instead of using $\omega$ as a wildcard. Second, in order to describe a sequence in which $d'$ is last and $d$ is $n-1$ tiles before $d$, i.e., $d'$ is vertically adjacent to $d$ in a tiling, instead of using a type such as $d' \to \omega \to \ldots \omega \to d \to \alpha$ we use $d' \to \gamma_{d}^n$. By a simple inductive proof we have that for substitution $S$ satisfying constraints ($iv$) and ($v$) if $S(d' \to \gamma_d^i) \leq d_1 \to \ldots \to d_k \to \bullet$, then $d_1 = d'$ and $d_{n+1} = d$. Additionally, intersections in $\gamma_d^i$ appear either on the top level or in arrow target position. Therefore, using distributivity, it still suffices to consider only substitutions that map variables to intersections of types of the shape $d_1 \to \ldots \to d_k \to \bullet$ where $d_1, \ldots, d_k \in D$.

\begin{theorem}
\label{thm:lb_no_omega}
The intersection type satisfiability problem in the $\omega$-free fragment and the intersection type unification problem in the $\omega$-free fragment are \textsc{Exptime}-hard.
\end{theorem}

\begin{proof}
Same as the proof of Theorem \ref{thm:lb} using the set of constraints $\mathscr{C}_T'$.
\end{proof}

\subsection{Alternative Presentation}
\label{sec:alternative-presentation}

The most peculiarly shaped axiom in $\textsc{ACIUD}_l\textsc{ReAb}$ is absorption. Therefore, we would like to explore a different presentation of $\textsc{ACIUD}_l\textsc{ReAb}$ unification that does not rely on such oddly shaped axiom. For this purpose, we introduce the least upper bound $\vee$ (cf.~\cite[Theorem 2.9]{KT95}) as a syntactic meta-operator to intersection types. In particular, we can define $\vee : \TT \times \TT \to \TT$ such that for all $\sigma, \tau, \rho \in \TT$ we have $\sigma \leq \rho$ and $\tau \leq \rho$ iff $\sigma \vee \tau \leq \rho$. As a result, the axiom (\textsc{Ab}) is subsumed by lattice-theoretic operations. Note that using the meta-operator $\vee$ is different from using union types $\cup$, which are strictly more expressive. 

Considering the way the least upper bound operation is defined for intersection types, we split the original absorption axiom into lattice-theoretic absorption (AB$_\cap$) and contravariant right-distributivity (D$_{r}^-$) in the following Definition \ref{def:aciudlreab-vee}.

\begin{definition}[$\textsc{ACIUD}^-\textsc{ReAb}_\cap$] 
\label{def:aciudlreab-vee}
The equational theory $\textsc{ACIUD}^-\textsc{ReAb}_\cap$ is given by
\begin{description}[before={\renewcommand\makelabel[1]{\upshape\bf ##1}}]
\item[(A)] $\sigma\cap(\tau\cap\rho) \approx (\sigma\cap\tau)\cap\rho$
\item[(C)] $\sigma\cap\tau \approx \tau\cap\sigma$
\item[(I)] $\sigma\cap\sigma \approx \sigma$
\item[(U)] $\sigma \cap\omega \approx \sigma$
\item[(D$_{l}$)] $(\sigma\to\tau)\cap(\sigma\to\tau') \approx \sigma\to\tau\cap\tau'$
\item[(RE)] $\omega \approx \omega \to \omega$
\item[(AB$_\cap$)] $\sigma \approx \sigma \cap (\sigma \vee \tau)$
\item[(D$_{r}^-$)] $(\sigma\to\tau)\vee(\sigma'\to\tau) \approx (\sigma \cap \sigma')\to\tau$
\end{description}
\end{definition}

The theory $\textsc{ACIUD}^-\textsc{ReAb}_\cap$ interacts with $\vee$ syntactically. Since we are interested in $\vee$-free intersection types, we need to verify that $\approx$ is identical to $\sim$ on $\TT$.

\begin{lemma}
\label{lem:sim-is-approx}
Given $\sigma, \tau \in \TT$ we have $\sigma \sim \tau$ in $\textsc{ACIUD}_l\textsc{ReAb}$ iff $\sigma \approx \tau$ in $\textsc{ACIUD}^-\textsc{ReAb}_\cap$. 
\end{lemma}

\begin{proof} \enquote{$\pmb{\Longrightarrow}$}: The theory $\textsc{ACIUD}^-\textsc{ReAb}_\cap$ splits the original absorption axiom into lattice-theoretic absorption (AB$_\cap$) and contravariant right-distributivity (D$_{r}^-$) as follows
\begin{align*}
\sigma \to \tau &\stackrel{\bf AB_\cap}{\approx} (\sigma \to \tau) \cap ((\sigma \to \tau) \vee (\sigma' \to \tau)) 
\stackrel{\bf D_{r}^-}{\approx} (\sigma \to \tau) \cap (\sigma \cap \sigma' \to \tau)
\end{align*}
Therefore, $\sigma \sim \tau$ implies $\sigma \approx \tau$ by replacing (\textbf{AB}) as shown above.

\enquote{$\pmb{\Longleftarrow}$}: First, we show that the additional axioms (\textbf{AB}$_\cap$) and (\textbf{D}$_{r}^-$) are derivable using subtyping and the least upper bound property of $\vee$.
\begin{description}[before={\renewcommand\makelabel[1]{\upshape\bf ##1}}]
\item[(AB$_\cap$)] First, $\sigma \leq \sigma$ and $\sigma \leq \sigma \vee \tau$ implies $\sigma \leq \sigma \cap (\sigma \vee \tau)$. Second, $\sigma \cap (\sigma \vee \tau) \leq \sigma$. Therefore, $\sigma = \sigma \cap (\sigma \vee \tau)$.
\item[(D$_{r}^-$)] Due to~\cite[Theorem 2.9]{KT95} $\sigma \cap \sigma' \to \tau$ is the least upper bound of $\sigma \to \tau$ and $\sigma' \to \tau$. By least upper bound uniqueness, we have 
$$(\sigma\to\tau)\vee(\sigma'\to\tau) = (\sigma \cap \sigma')\to\tau$$
\end{description}
As a result, $\sigma \approx \tau$ implies $\sigma = \tau$, therefore $\sigma \sim \tau$ by Lemma \ref{lem:sim-is-eq}.
\end{proof}

\begin{corollary}
\label{lem:sim-is-eq-isapprox}
Given $\sigma, \tau \in \TT$ we have $\sigma = \tau$ iff $\sigma \approx \tau$.
\end{corollary}

\begin{remark}
Since $\cap$ is greatest lower bound wrt. $\leq$ and $\vee$ is least upper bound wrt. $\leq$, we may add lattice theoretic properties of $\vee$ (such as ACI and mutual distributivity of $\cap$ and $\vee$) to $\textsc{ACIUD}^-\textsc{ReAb}_\cap$ without changing its expressiveness.
\end{remark}
It remains open whether one can exploit additional lattice theoretic properties introduced by $\vee$ to bound unification complexity similar to the type tallying problem~\cite{CNXA15}.

\section{On Decidability of Rank 1 Unification}
\label{sec:rank1}

Since problems in the intersection type system are often sensitive to rank~\cite{urzy09}, it is interesting to inspect decidability of rank restricted unification. 

By \emph{rank 1 unification} (resp. \emph{rank 1 satisfiability}) we denote unification (resp. satisfiability) where the substitution codomain is restricted to rank 1 types. To be precise, each variable can be mapped to a (possibly empty, i.e. $\omega$) intersection of simple types (not containing intersections or $\omega$).

In this section we sketch the reduction of rank 1 unification to set constraints with projections in finite sets with cardinalities. 
Unfortunately, only the complexity of set constraints with projections in infinite sets is known to be \textsc{NExptime}~\cite{CharatonikP94} and complexity in finite sets is unknown. For an overview of properties of set constraints with projections see~\cite{CharatonikP94}.
The purpose of the sketched reduction is to provide a toolbox to translate the oddities of intersection type unification in a uniform set theoretic setting, raising the question regarding its decidability. Although intersection types have a set-theoretic interpretation in a $\lambda$-calculus model~\cite{bcd} where subtyping coincides with set inclusion at the level of {\em $\lambda$-terms}, it is not evident how the the shape of legal solutions (possibly infinite sets of $\lambda$-terms of an arbitrary type) can be expressed in the syntax of set constraints with projections in the sense of~\cite{CharatonikP94}.
The fundamental observation that distinguishes rank 1 unification is that rank 1 subtyping degenerates to set inclusion (cf. Lemma \ref{lem:simple-intersection-subtyping}) at the level of {\em types}.

Given a set $\calC$ of constraints, we non-deterministically transform $\calC$ to a set of set constrains with projections $\calD$ such that $\calC$ has a solution iff $\calD$ has a solution. The signature of the resulting set constraints is that of simple types, where $\src$ (resp. $\tgt$) is the first (resp. second) projection of $\to$ lifted to sets. The cardinality constraint $|\alpha| = 1$ implies that any solution requires the set substituted for $\alpha$ to contain exactly one element.

To construct $\calD$, first choose a set $V$ of variables and replace each occurrence of each variable in $V$ by $\omega$. Second, normalize (recursively organize) all occurring types. Third, iteratively apply the following rules until $\calC$ is empty (given by highest to lowest priority)


\begin{enumerate}[leftmargin=2.2em]
\item Remove $\sigma \dotleq \tau$ from $\calC$ if $\sigma \leq \tau$.
\item Remove $\varphi \dotleq \alpha$ from $\calC$ if $\varphi$ is a simple type and add $\alpha = \{\varphi\}$ to $\calD$ (cf. Lemma \ref{lem:simple-subtyping}).
\item Remove $\alpha \dotleq \varphi$ from $\calC$ if $\varphi$ is a simple type and add $\{\varphi\} \subseteq \alpha$ to $\calD$ (cf. Lemma \ref{lem:simple-intersection-subtyping}).
\item Remove $\alpha \dotleq \beta$ from $\calC$ and add $\beta \subseteq \alpha$ to $\calD$ (cf. Lemma \ref{lem:simple-intersection-subtyping}).
\item If $\omega \dotleq \tau$ in $\calC$ and $\tau \not\in \TT^\omega$ abort with \enquote{no solution} (cf Lemma \ref{lem:tt_omega}).
\item Replace $\sigma \dotleq \bigcap_{i \in I} \tau_i$ in $\calC$ by $\{\sigma \dotleq \tau_i \mid i \in I\}$.
\item If $a \dotleq \sigma \to \tau$, $\sigma \to \tau \dotleq a$ or $a \dotleq b$ in $\calC$ then abort with \enquote{no solution}.
\item Replace $\bigcap_{i \in I} \tau_i \dotleq \sigma_1 \to \ldots \to \sigma_n \to a$ in $\calC$ by $\tau_i \dotleq \sigma_1 \to \ldots \to \sigma_n \to a$ where $i \in I$ is chosen non-deterministically (cf. Corollary \ref{cor:path_subtyping}).
\item Replace $\bigcap_{i \in I} \tau_i \dotleq \sigma_1 \to \ldots \to \sigma_n \to \alpha$ in $\calC$ by $\{ \tau_i \dotleq \sigma_1 \to \ldots \to \sigma_n \to \alpha_i \mid i \in I' \}$ and add $\alpha = \bigcup_{i \in I'} \alpha_i$ to $\calD$ where $\emptyset \neq I' \subseteq I$ is chosen non-deterministically and $\alpha_i$ for $i \in I'$ are fresh (cf. Lemma \ref{lem:organized_subtyping}).
\item Replace $\sigma_1 \to \tau_1 \dotleq \sigma_2 \to \tau_2$ in $\calC$ by $\{\sigma_2 \dotleq \sigma_1, \tau_1 \dotleq \tau_2\}$.
\item Replace $\alpha \dotleq \sigma \to \tau$ in $\calC$ by $\{\sigma \dotleq \beta, \gamma \dotleq \tau \}$ and add $\{\delta \subseteq \alpha, \beta = \src(\delta), \gamma = \tgt(\delta)\}$ to $\calD$ where $\beta, \gamma, \delta$ are fresh (cf. Lemma \ref{lem:arrow_subtyping}). 
\item Replace $\omega \to \tau \dotleq \alpha$ in $\calC$ by $\tau \dotleq \beta$ and add $\alpha \subseteq \gamma \to \beta$ to $\calD$ where $\beta, \gamma$ are fresh.
\item Replace $\bigcap_{i \in I} \sigma_i \to \tau \dotleq \alpha$ in $\calC$ by $\{\sigma_i \to \tau \dotleq \alpha \mid i \in I\}$ (cf.~\cite[Theorem 2.9]{KT95}).
\item Replace $(\sigma_1 \to \ldots \to \sigma_n \to a) \to \tau \dotleq \alpha$ in $\calC$ by $\{\sigma_i \dotleq \beta_i \mid i = 1 \ldots n\} \cup \{\tau \dotleq \gamma\}$ and add $\alpha \subseteq (\beta_1 \to \ldots \to \beta_n \to a) \to \gamma$ to $\calD$ where $\beta_i, \gamma$ are fresh for $i = 1 \ldots n$ (cf. Lemma \ref{lem:simple-contravariant-subtyping}).
\item \label{constr:case-arrow}
Replace $(\sigma_1 \to \ldots \to \sigma_n \to \delta) \to \tau \dotleq \alpha$ in $\calC$ by $\{\sigma_i \dotleq \beta_i \mid i = 1 \ldots n\} \cup \{\tau \dotleq \gamma\}$ and add $\{\alpha \subseteq (\beta_1 \to \ldots \to \beta_n \to \delta) \to \gamma, |\delta| = 1\}$ to $\calD$ where $\beta_i, \gamma$ are fresh for $i = 1 \ldots n$ (cf. Lemma \ref{lem:simple-contravariant-subtyping}).
\end{enumerate}

The above sketched construction contains justifications for its soundness and completeness.
A crucial aspect for rank 1 unification is that subtyping restricted to simple types degenerates to equality.

\begin{lemma}
\label{lem:simple-subtyping}
For $\varphi, \psi \in \TT$ such that $\varphi$ and $\psi$ are simple types we have $\varphi \leq \psi$ iff $\varphi \equiv \psi$.
\end{lemma}

\begin{proof}
By induction on the sum of depths of $\varphi$ and $\psi$ using Lemma \ref{lem:beta_soundness}.
\end{proof}

Since a simple type is also a path (cf. Definition \ref{def:path}), subtyping restricted to intersections of simple types degenerates to set inclusion.

\begin{lemma}
\label{lem:simple-intersection-subtyping}
For $\sigma, \tau \in \TT$ such that $\sigma \equiv \bigcap_{i \in I} \varphi_i$ and $\tau \equiv \bigcap_{j \in J} \psi_j$ are intersections of simple types we have $\sigma \leq \tau$ iff $\{\psi_j \mid j \in J \} \subseteq \{\varphi_i \mid i \in I\}$.
\end{lemma}

\begin{proof}
Consequence of Lemma \ref{lem:simple-subtyping} and Lemma \ref{lem:organized_subtyping}.
\end{proof}

An other crucial property of intersection type subtyping that allows to strictly strengthen inequalities is expressed by the following Lemma \ref{lem:arrow_subtyping}.

\begin{lemma}
\label{lem:arrow_subtyping}
We have $\bigcap_{i \in I}(\sigma_i \to \tau_i) \leq \sigma' \to \tau'$ iff there exists an index set $I' \subseteq I$ such that $\bigcap_{i \in I'} \sigma_i \to \bigcap_{i \in I'}\tau_i \leq \sigma' \to \tau'$.
\end{lemma}

\begin{proof}
Immediate consequence of Lemma \ref{lem:beta_soundness} $(iii)$.
\end{proof}

Finally, the most important case $(\ref{constr:case-arrow})$ in the above construction relies on a more involved argument. In particular, we need to restrict the cardinality of solutions due to the following Lemma \ref{lem:simple-contravariant-subtyping}.

\begin{lemma}
\label{lem:simple-contravariant-subtyping}
We have $(\sigma_1 \to \ldots \to \sigma_n \to \bigcap_{i \in I} \varphi_i) \to \tau \leq \bigcap_{j \in J} \psi_j$ where $\varphi_i$, $\psi_j$ are simple types for $i \in I \neq \emptyset$, $j \in J \neq \emptyset$ iff for each $j \in J$ we have $\psi_j \equiv (\psi_j^1 \to \ldots \to \psi_j^n \to \psi_j^{n+1}) \to \psi_j^{n+2}$ for some $\psi_j^1, \ldots, \psi_j^{n+2}$ and the following conditions hold
\begin{enumerate}[label=(\roman*)]
\item $\tau \leq \bigcap_{j \in J} \psi_j^{n+2}$
\item $\sigma_k \leq \bigcap_{j \in J} \psi_j^k$ for $k = 1 \ldots n$
\item $\phi_i \equiv \psi_j^{n+1}$ for all $i \in I$ and $j \in J$
\end{enumerate}
\end{lemma}

\begin{proof}
\enquote{$\pmb{\Longrightarrow}$}: For each $j \in J$ the shape $\psi_j \equiv (\psi_j^1 \to \ldots \to \psi_j^n \to \psi_j^{n+1}) \to \psi_j^{n+2}$ for some $\psi_j^1, \ldots, \psi_j^{n+2}$ is due to Lemma \ref{lem:beta_soundness} with the resulting conditions
\begin{enumerate}
\item for all $j \in J$ we have $\tau \leq \psi_j^{n+2}$
\item for all $j \in J$ and $k = 1 \ldots n$ we have $\sigma_k \leq \psi_j^k$
\item for all $j \in J$ we have $\psi_j^{n+1} \leq \bigcap_{i \in I} \varphi_i$
\end{enumerate}
Conditions $(1)$ and $(2)$ are equivalent to conditions $(i)$ and $(ii)$. By Lemma \ref{lem:simple-intersection-subtyping} we have $\{\varphi_i \mid i \in I\} \subseteq \{\psi_j^{n+1}\}$ for each $j \in J$ which implies condition $(iii)$.

\enquote{$\pmb{\Longleftarrow}$}: Immediately by definition of $\leq$.
\end{proof}

\section{Conclusion and future work}
\label{sec:conclusion}
We have positioned the algebraic intersection type unification problem as a natural object of study within unification theory and type theory, and
we have provided the first nontrivial lower bound in the standard and the $\omega$-free fragment showing that the problem is of high complexity.
Our \textsc{Exptime}-lower bound uses game-theoretic methods which may be useful for making further progress on the main open question for future work, that of decidability. 
Next steps include exploring variants and restrictions. Variants of intersection type subtyping theories (see \cite{BDS13}) give rise to a {\em family} of intersection type unification problems yet to be studied.
We conjecture an \textsc{NExptime}-upper bound for rank 1 restricted unification, in which variables are substituted by intersections of simple types. Since organized rank 1 subtyping corresponds to set inclusion, one can reduce a rank 1 unification problem to satisfiability of set constraints with projections~\cite{CharatonikP94} in finite sets. Unfortunately, 
standard set constraint interpretations may contain infinite sets, which is why this approach needs further investigation.


\section*{Acknowledgement}
We would like to thank the following colleagues for helpful discussions:
Boris D\"{u}dder (Dortmund and Copenhagen), Pawe\l\ Urzyczyn, Aleksy Schubert, and Marcin Benke (Warsaw), 
and Mariangiola Dezani, Simona Ronchi Della Rocca, Mario Coppo, Ugo de'Liguoro, and the Torino $\lambda$-calculus group.
Thanks are also due to our reviewers for useful comments.


\bibliographystyle{alpha}
\bibliography{../../bibliographyLS14}

\end{document}